\documentclass[journal]{IEEEtran}
\usepackage{stmaryrd}
\usepackage{subfig}
\usepackage{xspace}
\usepackage{url}
\usepackage[overload]{empheq}
\usepackage{cleveref} 
\usepackage{amsmath}
\usepackage{amssymb}
\usepackage{amsthm}
\usepackage{cite}
\usepackage[font={small}]{caption}

\newcommand{\link}{\ensuremath{(i,j)}\xspace}
\newcommand{\lien}[2]{\ensuremath{(#1,#2)}\xspace}
\newcommand{\ensP}{\ensuremath{\mathcal{E}_P}\xspace}
\newcommand{\ensNP}{\ensuremath{\mathcal{E}_{NP}\xspace}}
\newcommand{\ensA}{\ensuremath{\mathcal{E}_{A}\xspace}}

\DeclarePairedDelimiter\ceil{\lceil}{\rceil}

\usepackage{tikz}
\usetikzlibrary{shapes}
\usetikzlibrary{graphs}
\usetikzlibrary{calc}

\newcommand{\set}[1]{\ensuremath{\mathcal{S}_#1}}
\newcommand{\crit}[1]{\ensuremath{k_{#1}}}
\newcommand{\card}[1]{\ensuremath{n_{#1}}}
\usepackage{tkz-berge}
\newtheorem{proposition}{Proposition}
\newtheorem{definition}{Definition}
\newtheorem{corollary}{Corollary}

\title{Optimal Secure Multi-Layer IoT Network Design}

\author{Juntao Chen,~\IEEEmembership{Student Member,~IEEE,} Corinne Touati, and Quanyan Zhu,~\IEEEmembership{Member,~IEEE} 
\thanks{J. Chen and Q. Zhu are with the Department of Electrical and Computer Engineering, Tandon School of Engineering, New York University, Brooklyn, NY 11201, USA. email: \{jc6412,qz494\}@nyu.edu}
\thanks{
         C. Touati is with INRIA, F38330 Montbonnot Saint-Martin, France. email: corinne.touati@inria.fr}%

\thanks{This research is partially supported by a DHS grant through Critical Infrastructure Resilience Institute (CIRI), grants ECCS-1847056, CNS-1544782 and SES-1541164 from National Science of Foundation (NSF), and grant W911NF-19-1-0041 from ARO.}
}

\begin{document}

\bstctlcite{IEEEexample:BSTcontrol}

\maketitle
\thispagestyle{empty}
\pagestyle{empty}

\tikzset{protege/.style = {line width=4pt}}
\tikzset{attaque/.style = {style=loosely dashed}}
\tikzset{both/.style = {line width=3pt, style=loosely dashed}}
\tikzset{multi/.style = {
shape = rectangle,
line width=1pt,
inner sep = 1.5pt,
draw
}}

\begin{abstract}
With the remarkable growth of the Internet and communication technologies over the past few decades, Internet of Things (IoTs) is enabling the ubiquitous connectivity of heterogeneous physical devices with software, sensors, and actuators. IoT networks are naturally two-layer with the cloud and cellular networks coexisting with the underlaid device-to-device (D2D) communications. The connectivity of IoTs plays an important role in information dissemination for mission-critical and civilian applications. However, IoT communication networks are vulnerable to cyber attacks including the denial-of-service (DoS) and jamming attacks, resulting in link removals in IoT network.  In this work, we develop a heterogeneous IoT network design framework in which a network designer can add links to provide additional communication paths between two nodes or secure links against attacks by investing resources. By anticipating the strategic cyber attacks, we characterize the optimal design of secure IoT network by first providing a lower bound on the number of links a secure network requires for a given budget of protected links, and then developing a method to construct networks that satisfy the heterogeneous network design specifications. Therefore, each layer of the designed heterogeneous IoT network is resistant to a predefined level of malicious attacks with minimum resources. Finally, we provide case studies on the Internet of Battlefield Things (IoBT) to corroborate and illustrate our obtained results.
\end{abstract}

\begin{IEEEkeywords}
Optimal Design, Two-Layer Networks, Security, Connectivity, Internet of Battlefield Things
\end{IEEEkeywords}

\section{Introduction}
 Internet of Things (IoTs) have witnessed a tremendous development with a variety of applications, such as virtual reality,  intelligent supply chain \cite{Pang} and smart home \cite{Wang}. In this highly connected world, IoT devices are massively deployed and connected to cellular or cloud networks. 
For example, in smart grids, wireless sensors are adopted to collect the data of buses and power transmission lines. The collected data can then be sent to a supervisory control and data acquisition (SCADA) center through cellular networks for grid monitoring and decision planning purposes. 
Smart home is another example of IoT application. Various devices and appliances in a smart home including air conditioner, lights, TV, tablets, refrigerator and smart meter are interconnected through the cloud, improving the quality of the living.

IoT networks can be viewed as multi-layer networks with the existing infrastructure networks (e.g., cloud and cellular networks) and the underlaid device networks. 
The connections between different objects in the IoT network can be divided into two types. 
Specifically, the communications between devices themselves are called \textit{interlinks}, while the devices communicate with the infrastructure through \textit{intralinks}.
 The connectivity of IoT networks plays an important role in information dissemination. On the one hand, devices can communicate directly with other devices in the underlaid network for local information. On the other hand, devices can also communicate with the infrastructure networks to maintain a global situational awareness. In addition, for IoT devices with insufficient on-board computational resources such as wearables and drones, they can outsource heavy computations to the data centers through cloud networks, and hence extend the battery lifetime.
Vehicular network is an illustrative example for understanding the two-tier feature of IoT networks \cite{Dey}. In an intelligent transportation network, vehicle-to-vehicle (V2V) communications enable two vehicles to communicate and exchange information, e.g., accidents, speed alerts, notifications. In addition, vehicles can also communicate with roadside infrastructures or units (RSU) that belong to one or several service providers for exchanging various types of data related to different applications including GPS navigation, parking and highway tolls inquiry. In this case, the vehicles form one network while the infrastructure nodes form another network. Due to the interconnections between two networks, vehicles can share information through infrastructure nodes or by direct V2V communications. 

IoT communication networks are vulnerable to cyber attacks including the denial-of-service (DoS) and jamming attacks \cite{Abomhara}. To compromise the communication between two specific devices, the attacker can adopt the selective jamming attack \cite{law2009energy,lazos2011selective}.  More specifically, the attacker selectively targets specific channels and packets which disrupts the communications by transmitting a high-range or high-power interference signal. This adversarial behavior leads to communication link removals in IoT network. Therefore, to maintain the connectivity of devices, IoT networks need to be secure and resistant to malicious attacks. For example, V2V communication links of a car can be jammed, and hence the car loses the real-time traffic information of the road which may further cause traffic delays and accidents especially in the futuristic self-driving applications. Hence, IoT networks should be constructed in a tactic way by anticipating the cyber attacks. Internet of Battlefield Things (IoBT) is another example of mission-critical IoT systems. As depicted in Fig. \ref{battlefield}, in IoBT networks, a team of unmanned aerial vehicles (UAVs) serves as one layer of wireless relay nodes for a team of unmanned ground vehicles (UGVs) and soldiers equipped with wearable devices to communicate between themselves or exchange critical information with the command-and-control nodes. The UAV network and the ground network naturally form a two-layer network in a battlefield which can be susceptible to jamming attacks. It is essential to design communication networks that can allow the IoBT networks to be robust to natural failures and secure to cyber attacks in order to keep a high-level situational awareness of agents in a battlefield. 

\begin{figure}[!t]
\begin{centering}
\includegraphics[width=1\columnwidth]{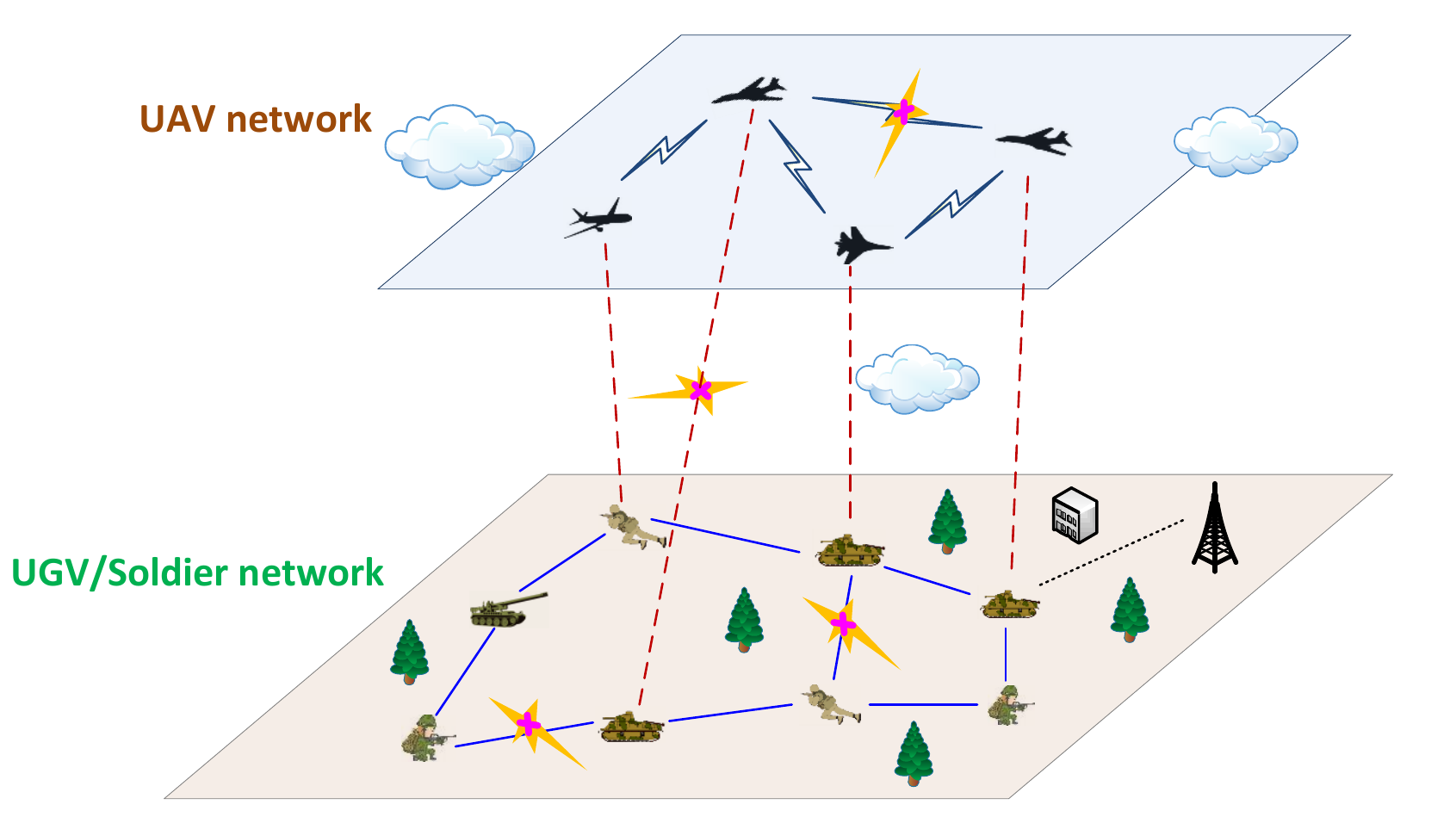} 
\par\end{centering}
\caption{\label{battlefield}
In IoBT networks, a team of UAVs and a group of soldiers and UGVs execute missions cooperatively. The agents in the battlefield share critical information through D2D communications. The UAV network and  ground network form a two-layer network which faces cyber threats, e.g., jamming attacks which can lead to link removals.}
\end{figure}

Due to heterogeneous and multi-tier features of the IoT networks, the required security levels can vary for different networks. For example, in IoBT networks, the connectivity of UAV networks requires a higher security level than the ground network if the UAVs are more likely to be targeted by the adversary. Similarly, in vehicular networks, the communication links between RSUs need a high-level protection when they anticipate more attacks than the vehicles do.  Therefore, it is imperative to design secure IoT networks resistant to link attacks and maintain the two-layer network connectivity with heterogeneous security requirements simultaneously.
 To this end, we present a heterogeneous IoT network design framework in which network links are vulnerable to malicious attacks. To enhance the security and the robustness of the network, an IoT network designer can add extra links to provide additional communication paths between two nodes or secure links against failures by investing resources to protect the links. To allocate links, note that when the nodes in the IoT network are within a short distance, then the classical wireless communication technologies can be adopted including WiFi, Bluetooth, and Zigbee. In comparison, when the distance is large, then one option that has recently emerged is called ultra narrow band (UNB) \cite{lassen2014long} that uses the random frequency and time multiple access \cite{li20172d}. The UNB is dedicated for mission-critical IoT systems for providing reliable communication services in long range. The goal of the multi-tier network design is to make the network connectivity resistant to link removal attacks by anticipating the worst attack behaviors. 
Different from previous works \cite{Marcin,Christophe} which have focused on the secure design of single-layer networks, in our current work, the network designer needs to take into account the heterogeneous features of the IoT networks by imposing different security requirements on each layer which presents a new set of challenges for network design. 

In this paper, we focus on a two-layer IoT network and aim to design each network resistant to different number of link failures with minimum resources.  We characterize the optimal strategy of the secure network design problem by first developing a lower bound on the number of links a secure network requires for a given budget of protected links. Then, we provide necessary and sufficient conditions under which the bounds are achieved and present a method to construct an optimal network that satisfies the heterogeneous network design specifications with the minimum cost. Furthermore, we characterize the robust network topologies which optimally satisfy a class of security requirements. These robust optimal networks are applicable to the cases when the cyber threats are not perfectly perceived or change dynamically, typically happening in the mission-critical scenarios when the attacker's action is partially observable.

Finally, we use IoBT as a case study to illustrate the analytical results and obtain insights in designing secure networks. We consider a mission-critical battlefield scenario in which the UAV network anticipates higher cyber threats than the soldier network, and the number of UAVs is less than the number of soldiers. We observe that as the cost of forming a protected communication link becomes smaller, more secure connections are formed in the optimal IoBT network. In addition, the designed network is resilient to the change of agents in the battlefield. 
We also study the reconfiguration and resilience of the UAV network as nodes leave and join the battlefield.

The main contributions of this paper are summarized as follows:
\begin{enumerate}
\item We propose a two-layer heterogeneous framework for IoT networks consisting of various devices, where each layer network faces different levels of cyber threats.
\item By utilizing the tools from graph theory and optimization, we analyze the lower bounds of the number of required links for the IoT network being connected by anticipating the worst case attacks.
\item We derive optimal strategies for creating secure two-layer IoT networks with heterogeneous security requirements and provide their construction guidelines under different regimes in terms of threat levels and number of nodes. We also identify the robust optimal strategies for the IoT network with dynamic cyber threat levels.
\item We apply the optimal design principles to crucial IoBT scenarios and provide insights into the design of secure and resilient interdependent UAV and soldier networks in the IoBT. 
\end{enumerate}

\subsection{Related Work}
Due to the increasing cyber threats, IoT security becomes a critical concern nowadays \cite{Weber}. Depending on the potential of cyber attackers, IoT networks face heterogeneous types of attacks \cite{Nia}. For example, attackers can target the edge computing nodes in IoT, e.g., RFID readers and sensor nodes. Some typical adversarial scenarios include the node replication attack by replicating one node's identification number \cite{Parno}, DoS by battery draining, sleep deprivation, and outage attacks \cite{Khouzani,Vasserman}. The attackers can also launch attacks through the IoT communication networks. Quintessential  examples include the eavesdropping attack where the attacker captures the private information over the channel, and utilizes the information to design other tailored attacks \cite{Mukherjee}. Another example is the data injection attack where the attacker can inject fraudulent packets into IoT communication links through insertion, manipulation, and replay techniques \cite{Zhou4625802}.  In our work, we focus on the jamming and DoS attacks which lead to the link removal in IoT communication networks.

To mitigate the cyber threats in IoT, a large number of works have focused on addressing the security issues by using different methodologies \cite{Abomhara}.  
A contract-theoretic approach has been adopted to guarantee the performance of security services in the Internet of controlled things \cite{chen2016optimal,chen2017} and mitigate the systemic cyber risks \cite{chen2018linear}.
 The authors in \cite{Zhou} have proposed a media-aware security architecture for facilitating multimedia applications in the IoT. \cite{chen2017dynamic} has proposed a dynamic game model including pre-attack defense and post-attack recovery phases in designing resilient IoT-enabled infrastructure networks.  Strategic security investment under bounded rationality in IoT has been studied in \cite{chen2018security,chen2019-TIFS}. The authors in \cite{pawlick2019istrict} have developed an interdependent strategic trust mechanism to defend against cyber attacks in IoT.

In this work, we investigate the secure design of IoT network by considering its connectivity measure \cite{Marcin,Christophe,Chen_ACC,Chen_CDC} through the lens of graph theory \cite{gross2004handbook}. Comparing with the previous works \cite{Marcin,Christophe} that have focused on a single-layer adversarial network design, we model the IoT as a two-layer network and strategically design each layer of the network with heterogeneous security requirements. The current work is also related to the secure and resilient interdependent critical infrastructures \cite{huang2017large,huang2017factored,huang2018distributed,
huang2018factored,chen2017interdependent} in which a holistic design approach is required.

\subsection{Organization of the Paper} 
The rest of the paper is organized as follows. Section \ref{Two_Layer} formulates the heterogeneous two-layer IoT network design framework. Analytical results including the lower bounds of links and optimal IoT network design strategies are presented in Section \ref{optimal_strategy}. Case studies of IoBT networks are provided in Section \ref{case_study}, and Section \ref{conclusion} concludes the paper.

\section{Heterogeneous Two-Layer IoT Network\\ Design Formulation}\label{Two_Layer}
In this section, we formulate a two-layer secure IoT network design problem. Due to the heterogeneous features of IoT networks, the devices at each layer face different levels of cyber threats. To maintain the global situational awareness, the designer aims to devise an IoT  network with a minimum cost, where each layer of IoT network should remain connected in the presence of a certain level of adversarial attacks.

Specifically, we model the two-layer IoT network with two sets of devices or nodes\footnote{Nodes and vertices in the IoT network refer to the devices, and they are used interchangeably. Similar for the terms edges and links.} denoted by \set{1} and \set{2}.  Each set of nodes is of a different type. Specifically, denote by $\card{1} := |\set{1}|$ and $\card{2} := |\set{2}|$ the number of nodes of type $1$ and $2$, respectively, where $|\cdot|$ denotes the cardinality of a set. We unify them to $n = \card{1}+\card{2}$ \emph{vertices} that are numbered from $1$ to $n$ starting from nodes in \set{1}. Thus, a node labeled $i$ is of type $1$ if and only if $i \leq \card{1}$. Note that each set of nodes forms an IoT subnetwork. Together with the interconnections between two sets of nodes, the subnetworks form a two-layer IoT network. Technically, the communication protocols between nodes within and across different layers can be either the same or heterogeneous depending on the adopted technology by considering the physical distance constraints. Furthermore, the nodes' functionality can be different in two subnetworks depending on their specific tasks. In this paper, our focus lies in the high-level of network connectivity maintenance.

In standard graph theory, an \emph{edge} (or a \textit{link}) is an unordered pair of vertices: $(i,j) \in \llbracket 1, n \rrbracket^2$, $i\neq j$, where $\llbracket 1, n \rrbracket^2$ is a set including all the pairs of integers between 1 and $n$. We recall that two vertices (nodes) $i_0$ and $i_L$ are said {\it connected} in a graph of nodes $\set{1} \cup \set{2}$ and a set of edges $\mathcal{E}$ if there exists a path between them, i.e., a finite alternating sequence of nodes and distinct links: $i_0, \, (i_0,i_1), \, i_1, (i_1,i_2), \, i_2, \, ..., \, (i_{L-1},i_L),\,  i_L$, where $i_l \in \set{1} \cup \set{2}$ and $(i_{l-1}, i_{l}) \in \mathcal{E}$ for all $1 \leq l \leq L$.

In our IoT networks, the communication links (edges) are vulnerable to malicious attacks, e.g., jamming and DoS, which result in link removals. To keep the IoT network resistant to cyber attacks, the network designer can either invest (i) in redundancy of the path, i.e., using extra links so that two nodes can communicate through different paths, or (ii) in securing its links against failures where we refer to these special communication edges as \emph{protected links}. These protected links can be typically designed using moving target defense (MTD) strategies, where the designer randomizes the usage of communication links among multiple created channels between two nodes \cite{zhu}. More precisely, we consider that for the designer, the cost per non-protected link created is $c_{NP}$ and the cost per protected link created is $c_P$. It is natural to have $c_{NP} \leq c_P$ since creation of a protected link is more costly than that of a non-protected one. For clarity, we assume that the costs of protected or non-protected links at two different layers are the same. If the costs of creating links are different in two subnetworks, then the network designer needs to capture this link creation difference in his objective \cite{chen2017heterogeneous}. Let $\ensNP\subseteq \mathcal{E}$ be the set of non-protected links and $\ensP\subseteq \mathcal{E}$ be the set of protected links in the IoT network, and $\ensNP \cup \ensP = \mathcal{E}$. In this work, we assume that the protection is perfect, i.e., links will not fail under attacks if they are protected.  Therefore, an adversary does not have an incentive to attack protected links. Denote the strategy of the attacker by $\ensA$, then it is sufficient to consider attacks on a set of links $\ensA \subseteq \ensNP$. Furthermore, we assume that the network designer can allocate links between any nodes in the network. In the scenarios that setting up communication links between some nodes is not possible, then the network designer needs to take into account this factor as constraints when designing networks.

The heterogeneous features of IoT networks naturally lead to various security requirements for devices in each subnetwork. Hence, we further consider that the nodes in IoT network have different criticality levels (\crit{1} and \crit{2} for nodes of type $1$ and $2$, respectively, with $\crit{1}, \crit{2} \in \llbracket 0, |\ensNP| \rrbracket$, where $\llbracket a,b \rrbracket$ denotes a set of  integers between $a$ and $b$). It means that subnetworks 1 and 2 should remain connected after the compromise of \textit{any} \crit{1} and \crit{2} links in $\ensNP$, respectively. Thus, the designer needs to prepare for the worst case of link removal attacks when designing the two-layer IoT network. Our problem is beyond the robust network design where the link communication breakdown is generally caused by nature failures. In this paper, we consider the link removal which is a consequence of cyber attacks, e.g., jamming and DoS attack. Furthermore, in our problem formulation, the network designer can allocate protected links which can be seen as a security practice, and he takes into account the strategic behavior of attackers, and designs the optimal secure networks. Without loss of generality, we have the following two assumptions:
\begin{enumerate}
\item[(A1)] $\crit{1} \leq \crit{2}$.
\item[(A2)] $n_1 \geq 1$, $n_2 \geq 1$.
\end{enumerate}
 Specifically, (A1) indicates that the IoT devices in subnetwork 2 are relatively more important than those in subnetwork 1, and thus subnetwork 2 should be more resistant to cyber attacks. Another interpretation of (A1) can also be that subnetwork 2 faces a higher level of cyber threats, and the network designer needs to prepare a higher security level for subnetwork 2. In addition, (A2) ensures that no IoT subnetwork is empty.

More precisely, consider a set of vertices $\set{1} \cup \set{2}$ and edges $\ensP \cup \ensNP$. 
 The IoT network designer needs to guarantee the following two cases:
\begin{itemize}
\item[(a)]  if $|\ensA| \leq \crit{1}$, then all nodes remain attainable in the presence of attacks, i.e., $\forall i, j \in \set{1} \cup  \set{2}$, there exists a path in the graph $(\set{1} \cup \set{2},\ensP \cup \ensNP\backslash \ensA)$ between $i$ and $j$.
\item[(b)] if $|\ensA| \leq \crit{2}$, nodes of type $2$ remain attainable after attacks, i.e., $\forall i, j \in \set{2}$, there exists a path in the graph $(\set{1} \cup \set{2},\ensP \cup \ensNP\backslash \ensA)$ between $i$ and $j$.
\end{itemize}

\textit{\textbf{Remark:}} We denote the designed network satisfying (a) and (b) above by $s^D := (\set{1} \cup \set{2}, \ensP \cup \ensNP)$, and call such heterogeneous IoT networks  $(\crit{1}, \crit{2})$-\textit{resistant} (with $k_1 \leq k_2$). The proposed $(k1,k2)$-resistant metric provides a flexible network design guideline by specifying various security requirements on different network components. Furthermore, in this work, we care about each node's degree which requires an explicit agent-level quantification. Then, the $(k_1,k_2)$-resistant metric is more preferable than measure of the proportion of links in each subnetwork, where the latter metric only gives a macroscopic description of the link allocation over two subnetworks.

Given the system's parameters \set{1}, \set{2}, \crit{1}, and \crit{2},  an optimal strategy for the IoT network designer is the choice of a set of links $\ensP \cup \ensNP$ which solves the optimization problem:
\begin{align*}
\min_{\ensP,\ensNP}\quad &c_p |\ensP| + c_{NP}|\ensNP|\\
\text{s.t.}\ \ &\ensP \subseteq \llbracket 1, n \rrbracket^2, \ensNP \subseteq \llbracket 1, n \rrbracket^2, \\
&\ensP \cap \ensNP = \emptyset,\\
& s^D = (\set{1} \cup \set{2}, \ensP \cup \ensNP)\ \mathrm{is}\ (\crit{1}, \crit{2})\mathrm{-resistant}.
\end{align*}
From the above optimization problem, the optimal network design cost directly depends on $c_P$ and $c_{NP}$. In addition, as we will analyze in Section \ref{optimal_strategy}, the cost ratio $\frac{c_P}{c_{NP}}$ plays a critical role in the optimal strategy design.

Under the optimal design strategy, compromising a node with low degree, i.e., $k_1$ degree in subnetwork 1 and $k_2$ degree in subnetwork 2, is not feasible for the attacker, since the degree of any nodes without protected link in the network is larger than  \crit{1} or \crit{2} depending on the nodes' layers.

Note that the above designer's constrained optimization problem is not straightforward to solve. First, the size of search space increases exponentially as the number of nodes in the IoT network grows. Therefore, we need to find a scalable method to address the optimal network design. Second, the heterogeneous security requirements make the problem more difficult to solve. On the one hand, two subnetworks are separate since they have their own design standards. On the other hand, we should tackle these two layers of network design in a holistic fashion due to their natural couplings.

\begin{figure*}[htb]
\centering
\begin{tikzpicture}[line cap=round,line join=round, scale=0.12]
\draw [<->,thick] (0,50) node (yaxis) [right] {\small {\it Lower bound} on the number of non-protected links}
        |- (100,0) node [scale=1] (xaxis) [right] { $p$};
\newcommand\XX{45}
\newcommand\XA{30}\newcommand\YA{30}
\newcommand\XB{40}\newcommand\YB{20}
\newcommand\XC{80}\newcommand\YC{7}
\newcommand\XD{90}\newcommand\YD{0}

\node [scale=0.8]  (ZZ) at (-2,-2) {$0$};
\node [scale=0.8]  (A) at (\XA,-2) {$n_2-2$};
\node (H) [left] at (0,\XX) {$\frac{n_1(k_1+1)+n_2(k_2+1)}{2}$};
\node [scale=0.6]  (AA) at (\XA,-0.1) {};
\draw (\XA,-0.5) -- (\XA,0.5);
\node [thick,draw,circle,inner sep=2pt,fill] (ZZAAH) at (0,\XX) {};
\node (ZZHTexte) at (2,\XX+2) {A};
\node [thick,draw,circle,inner sep=2pt,fill] (AAH) at (\XA,\YA) {};
\node (AAHTexte) at (\XA+2,\YA+2) {B};
\node (AHA) [left] at (0,\YA) {$\frac{n_1(k_1+1)+2(k_2+1)}{2}$};
\node [scale=0.8]  (B) at (\XB,-2) {$n_2-1$};
\node [scale=0.6]  (BB) at (\XB,-0.1) {};
\draw (\XB,-0.5) -- (\XB,0.5);
\node [thick,draw,circle,inner sep=2pt,fill] (BBH) at (\XB,\YB) {};
\node (BBHTexte) at (\XB+2,\YB+2) {C};
\node (BHB) [left] at (0,\YB) {$\frac{(n_1+1)(k_1+1)}{2}$};
\node [scale=0.8]  (C) at (\XC,-2) {$n_1+n_2-2$};
\node [scale=0.4]  (CC) at (\XC,-1) {};
\draw (\XC,-0.5) -- (\XC,0.5);
\node [thick,draw,circle,inner sep=2pt,fill] (CCH) at (\XC,\YC) {};
\node (CCHTexte) at (\XC+2,\YC+2) {D};
\node [thick,draw,circle,inner sep=2pt,fill] (ZZCCH) at (\XD,\YD) {};
\node (DDHTexte) at (\XD+2,\YD+2) {E};
\node (CHC) [left] at (0,\YC) {${k_1+1}$};
\node [scale=0.8]  (D) at (\XD,-4) {$n_1+n_2-1$};
\node [scale=0.6]  (DD) at (\XC,-0.1) {};
\draw (\XC,-0.5) -- (\XC,0.5);

\draw [dotted] (AA) -- (AAH);
\draw [dotted] (BB) -- (BBH);
\draw [dotted] (CC) -- (CCH);

\draw [dotted] (AHA) -- (AAH);
\draw [dotted] (BHB) -- (BBH);
\draw [dotted] (CHC) -- (CCH);

\draw [line width=1.3pt, color=blue](ZZAAH) -- (AAH);
\draw [line width=1.3pt, color=black!45!green] (AAH) -- (BBH);
\draw [line width=1.3pt, color=brown] (BBH) -- (CCH);
\draw [line width=1.3pt, color=red] (CCH) -- (ZZCCH);

\draw [line width=1.3pt, dotted] (ZZAAH) -- (BBH);
\draw [line width=1.3pt, dotted] (BBH) -- (ZZCCH);


\node [color=blue] (AZE) at (20,40) {Slope $\frac{k_2+1}{2}$};
\node [color=black!45!green] (AZE) at (47,28) {Slope $\frac{(k_2+1)+(k_2-k_1)}{2}$};
\node [color=brown] (AZE) at (65,18) {Slope $\frac{k_1+1}{2}$};
\node [color=red] (AZE) at (93,7) {Slope $k_1+1$};
\end{tikzpicture}
\caption{Lower bound on the number of non-protected links as a function on the number of protected links in the IoT network. Note that all the slopes of lines are quantified in their absolute value sense for convenience.}\label{fig:illus} 
\end{figure*}
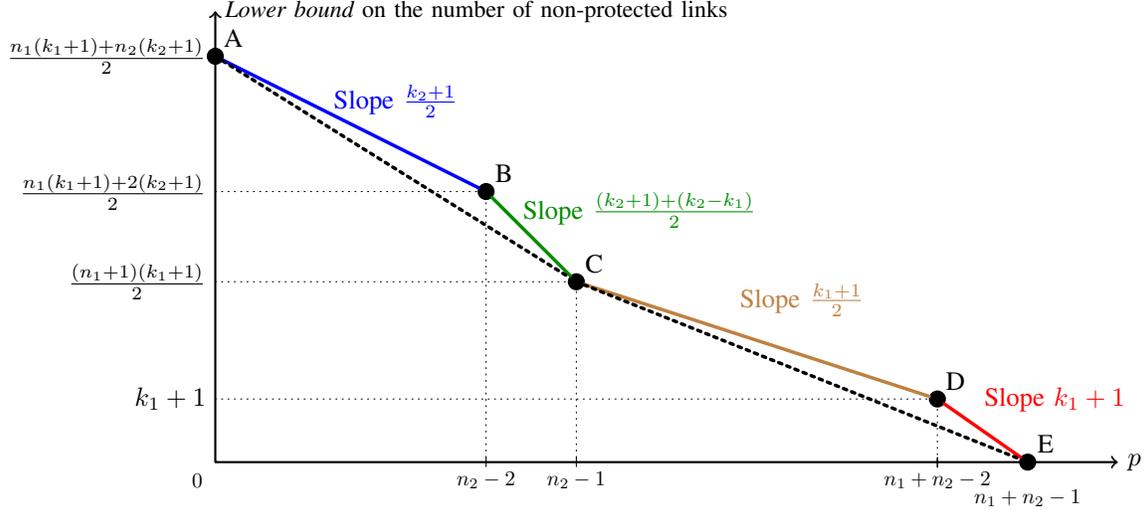

\section{Analytical Results and Optimal IoT Network Design}\label{optimal_strategy}
 
In this section, we provide an analytical study of the designer's optimal strategy, i.e., the optimal two-layer IoT network design.

We first develop, for given system parameters \set{1}, \set{2}, \crit{1}, \crit{2}, $c_P$ and $c_{NP}$, and for each possible number of protected links $p=|\ensP|$, a lower bound on the number of non-protected links that have any $(\crit{1}, \crit{2})$-resistant network with $p$ protected links (Section~\ref{sec:lower}). Then, we study three important cases, namely when $p$ takes values $0$, $n_2-1$ and $n_1+n_2-1$, and present for each of them sufficient conditions under which the lower bounds are attained (Section~\ref{sec:part}). Based on this study, we can obtain the main theoretical results of this paper, which include the optimal strategy for the designer, i.e., a $(\crit{1}, \crit{2})$-resistant IoT network with the minimal cost, as well as the robust  optimal strategy, and constructive methods of an optimal IoT network (Section~\ref{sec:theo}).

\subsection{A Lower Bound on the Number of (Non-Protected) Links}
\label{sec:lower}

Recall that the system parameters are \set{1}, \set{2}, \crit{1}, \crit{2}, $c_P$ and $c_{NP}$ (corresponding to the set of nodes of criticality level $1$ and $2$, the values of criticality, and the unitary cost of creating protected and non-protected links). 
We first address the question of a lower bound on the cost for the designer with an additional constraint on the number of protected links $p$ in the network. Since the cost is linear with the number of non-protected links, it amounts to finding a lower bound on the number of non-protected links that are required in any $(k_1, k_2)$-resistant network with $p$ protected links.

Let $\tilde{s}^D_p$ be a $(k_1, k_2)$-resistant network containing $p$ protected links. Then, we have the following proposition on the lower bound $|\ensNP|$.

\begin{proposition}[Lower bound on $|\ensNP|$]\label{prop:lower}
The number of non-protected links of $\tilde{s}^D_p$ is at least of
\begin{itemize}
\item[(i)] $\displaystyle \frac{n_1 (k_1+1) + (n_2-p)(k_2+1)}{2},$\hfill if $0\leq p \leq n_2-2$,
\item[(ii)] $\displaystyle \frac{(n-p) (k_1+1)}{2},$\hfill  if $n_2-1 \leq p \leq n_1+n_2-2$,
\item[(iii)] $0,$\hfill  if $ p=n_1+n_2-1$.
\end{itemize}
Note that $p$ takes integer values in each regime. The results are further illustrated in Fig.~\ref{fig:illus}.
\end{proposition}

Before proving Proposition \ref{prop:lower}, we first present the notion of network contraction in the following.

\textit{\textbf{Network Contraction:}} 
Let $g = (\set{1} \cup \set{2}, \ensP \cup \ensNP)$ be a network. Given a link $\link \in \ensP$, the network denoted by $g\oslash \link $ refers to the one obtained by contracting the link \link; i.e., by merging the two
nodes $i$ and $j$ into a single node $\{i,j\}$ (supernode). Note that any node $a$ is
adjacent to the (new) node $\{i,j\}$ in $g\oslash  \link$ if and only
if $a$ is adjacent to $i$ or $j$ in the original network $g$. 
In other words, all links, other than those incident to neither $i$ nor $j$, are links of $g\oslash  \link$ if and only if they are links of $g$.
Then $\hat{g}$, the contraction of network $g$, is the (uniquely defined) network obtained from $g$ by sequences of link contractions for all links in $\ensP$ \cite{Christophe}. 

For clarity, we illustrate the contraction of a network $g$ in Fig.~%
\ref{Figuredeb}. This example consists of $5$ nodes and $2$ protected links (represented in bold lines between nodes $1$ and $2$ and between nodes $3$ and $4$). The link $(1,2)$ is contracted and thus both nodes $1$ and $2$ in $g$ are merged into a single node denoted by $\{1,2\}$ in $\hat{g}$. Similarly the link $(3,4)$ is contracted. The resulting network thus consists of node $5$ and supernodes $\{1,2\}$ and $\{3,4\}$. Since $g$ contains a link between nodes $5$ and $1$ in $g$, then nodes $5$ and $\{1,2\}$ are connected through a link in network $\hat{g}$. Similarly, since nodes $1$ and $3$ are adjacent in $g$, then supernodes $\{1,2\}$ and $\{3,4\}$ are adjacent in network $\hat{g}$. 

\begin{figure}[t]
\centering{
\subfloat[Network $g = (N, \ensP \cup \ensNP)$]{
\phantom{aa}
\begin{tikzpicture}[x=1.1 cm,y=0.7 cm]
\GraphInit[vstyle=Normal]
\Vertex{5}
\NOEA(5){1} \EA(1){2} \SOEA(5){3} \EA(3){4}
\Edges(5,1,3) 
\tikzset{EdgeStyle/.style = protege}
\Edge(1)(2) \Edge(3)(4)
\end{tikzpicture}
\phantom{aa}
} \hspace{2em}
\subfloat[Contraction network $\hat{g}$]{
\begin{tikzpicture}[x=1.1 cm,y=0.7 cm]
\SetVertexMath
\GraphInit[vstyle=Normal]
\Vertex{5}
\presetkeys[GR]{vertex}{empty=false}{}
\tikzset{VertexStyle/.style = multi}
\NOEA[L={\{1,2\}}](5){3}
\SOEA[L={\{3,4\}}](5){pp}   
\Edge(5)(3)
\Edge(3)(pp)
\end{tikzpicture}\phantom{aaa}
}
\caption{Illustration of network contraction. The protected links $\lien{1}{2}$ and $\lien{3}{4}$ in network $g$ are contracted in network $\hat{g}$.}\label{Figuredeb}}
\end{figure}
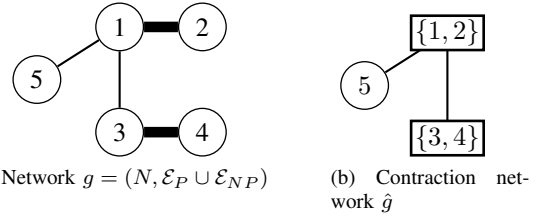

Based on network contraction, we present the proof of Proposition \ref{prop:lower} as follows.

\begin{proof}
Consider an IoT network $g$ including $p$ protected links, and $\hat{g}$ as its contraction. Let
\begin{itemize}
\item[(1)] $\nu_1$ be the number of nodes of type $1$ in $\hat{g}$ (and supernodes containing only nodes of type $1$),
\item[(2)] $\nu_2$ be the number of nodes of type $2$ in $\hat{g}$ (and supernodes containing only nodes of type $2$),
\item[(3)] $\nu_0$ be the number of supernodes in $\hat{g}$ that contains nodes of both type $1$ and $2$.
\end{itemize}

Note that if $\nu_1+\nu_2+\nu_0 = 1$, (i.e., if there is a unique supernode containing all nodes of the network), then no non-protected link is needed to ensure any level of $(\crit{1}, \crit{2})$-resistancy. 
Otherwise, for the IoT network to be $(k_1,k_2)$-resistant, each element of $\nu_1$, $\nu_2$ and $\nu_0$ must have a degree of (at least) $k_1+1$. 
Further, if there exist more than one element not in $\nu_1$; i.e., if $\nu_0+\nu_2 \geq 2$, then each of them should have a degree of (at least) $k_2+1$.

Thus, a lower bound on the number of non-protected links in $\tilde{s}^D_p$ is
\begin{equation*}
\Phi = \left\{\begin{array}{@{}l@{}}
\displaystyle \frac{\nu_1 (k_1+1) + (\nu_0+\nu_2)(k_2+1)}{2}, \quad \text{ if } \nu_2+\nu_0 > 1,\\
0, \hfill \text{if } \nu_1+\nu_2+\nu_0 = 1,\\
\displaystyle \frac{(\nu_1+1) (k_1+1)}{2}, \hfill \text{if } \nu_1 \geq 1 \text{ and }\nu_2+\nu_0 = 1.
\end{array}\right.
\end{equation*}

Next, we focus on the study of parameters $\nu_0$, $\nu_1$ and $\nu_2$. If no protected link is used, i.e., $p=0$, then $\nu_1 = n_1$, $\nu_2 = n_2$ and $\nu_0 = 0$ and $\nu_0+\nu_1+\nu_2 = n_1+n_2 = n$. Adding any protection allows to decrease the total number of elements $\nu_1+\nu_2+\nu_0$ by $1$ (or to remain constant if the link induce a loop in a protected component of $g$). Thus $\nu_0 + \nu_1 + \nu_2 \geq n-p$. Similarly, for each subnetwork, we have $\nu_0 + \nu_1 \geq n_1-p$ and $\nu_0 + \nu_2 \geq n_2-p$.
Further, the number of elements of $\nu_1$ and $\nu_2$ are upper bounded by  the number of nodes of type $1$ $n_1$ and type 2 $n_2$, respectively, i.e., $\nu_1 \leq n_1$ and $\nu_2 \leq n_2$. Finally, since $n_1 \geq 1$ then $\nu_1+\nu_0 \geq 1$, and since $n_2 \geq 1$ then $\nu_2+\nu_0 \geq 1$. 
Thus, for any $p$, a lower bound on the number of non-protected links in $\tilde{s}^D_p$ can be obtained by solving the following optimization problem:
\begin{equation} \label{eq:LP}
\begin{split}
\min_{\substack{\nu_1, \nu_2, \nu_0}}&\quad \Phi\\
\mathrm{s.t.}\quad & \nu_0 + \nu_1 + \nu_2 \geq n-p, \\
& \nu_0 + \nu_1 \geq n_1-p,\ \nu_0 + \nu_2 \geq n_2-p,\\ 
&\nu_1 \leq n_1, \; \nu_2 \leq n_2, \\
&\nu_1+\nu_0 \geq 1, \;\; \nu_2+\nu_0 \geq 1.
\end{split}
\end{equation}
To solve this optimization problem, we consider three cases.

\textit{Case 1:} First, assume that $p<n_2-1$. From $\nu_0 + \nu_1 + \nu_2 \geq n-p$, we obtain that $\nu_0+\nu_2 >1$. Thus, \eqref{eq:LP} reduces to $\min_{\substack{\nu_1, \nu_2, \nu_0}}\ \frac{\nu_1 (k_1+1) + (\nu_0+\nu_2)(k_2+1)}{2}$ with the same constraints as in \eqref{eq:LP} except $\nu_0+\nu_2 >1$.


Since $k_2 \geq k_1$, then the minimum of the objective is obtained when $\nu_0+\nu_2$ is minimized, i.e., when all protections involve nodes of type $2$. Then, $\nu_0+\nu_2 = n_2-p$. Thus, the lower bound is equal to $\frac{n_1 (k_1+1) + (n_2-p)(k_2+1)}{2}$. This result is illustrated by the line joining points A and B in Fig.~\ref{fig:illus}.

\textit{Case 2:} Assume that $n_2-1 \leq p \leq n_1+n_2-2$. Then $n-p \leq n_1+1$. Therefore, for a given $p$, i.e., for a given minimal value of $\nu_0+\nu_1+\nu_2$, we can have either $\nu_0+\nu_2>1$ or $\nu_0+\nu_2=1$. Then, the lower bound of the number of non-protected links is $\min\left\{ \frac{n_1(k_1+1)+(n_2-p)(k_2+1)}{2}, \frac{(n-p) (k_1+1)}{2}\right\}$. Recall that $k_2 \geq k_1$, and therefore the lower bound achieves at $ \frac{(n-p) (k_1+1)}{2} $. This observation is illustrated by the line in Fig.~\ref{fig:illus} joining points C and D.

\textit{Case 3:} Finally, when $p=n-1$,  $\nu_0+\nu_1+\nu_2 = 1$, and thus no non-protected link is needed, which is represented by point E in Fig.~\ref{fig:illus}.
\end{proof}

Based on Proposition \ref{prop:lower}, we further comment on the locations where protected and non-protected links are placed in the two-layer IoT networks.
\begin{corollary}\label{coro1}
When $0\leq p \leq n_2-2$, the protected links purely exist in subnetwork 2. When $n_2-1 \leq p \leq n_1+n_2-2$, subnetwork 2 only contains protected links, and non-protected links appear in subnetwork 1 or between two layers. When $p= n_1+n_2-1$, then all nodes in the two-layer IoT network are connected with protected links.
\end{corollary}

Corollary \ref{coro1} has a natural interpretation that the protected link resources are prior to be allocated to a subnetwork facing higher cyber threats, i.e., subnetwork 2 in our setting. 

\subsection{Networks with Special Values of $p$ Protected Links}
\label{sec:part}

In the previous Section \ref{sec:lower}, we have studied for each potential number of protected links $p$, a lower bound $m(p)$ on the minimum number of non-protected links  for an IoT network with sets of nodes $\set{1}$ and \set{2} being $(k_1, k_2)$-resistant. Then, the cost associated with such networks is $$C(p,m(p)) = p c_P + m(p) c_{NP},$$
where $C:\mathbb{N}\times \mathbb{N}\rightarrow \mathbb{R}_+$.
Since the goal of the designer is to minimize its cost, we need to investigate the value of $p$ minimizing such function $C(p,m(p))$. 

In Fig.~\ref{fig:illus}, we note that the plot of a network of equal cost (\textit{iso-cost}) $K$ is a line of equation $\frac{K-p c_P}{c_{NP}}$. It is thus a line of (negative) slope $c_P/c_{NP}$ that crosses the $y$-axis at point $K/c_{NP}$. Recall also that the graph that shows $m(p)$ as a function of $p$ is on the upper-right quadrant of its lower bound. Thus, the optimal value of $p$ corresponds to the point where an iso-cost line meets the graph $m(p)$ for the minimal value $K$. From the shape of the lower bound drawn in Fig.~\ref{fig:illus}, the points A, C and E are selected candidates leading to the optimal network construction cost. We thus investigate in the following the condition under which the lower bounds are reached at these critical points as well as the corresponding configuration of the optimal two-layer IoT networks.

\textbf{\textit{Remark:}} Denote by $s^D_p$ a (\crit{1}, \crit{2})-resistant IoT network with $p$ protected links and the \textit{minimum} number of non-protected links.

Before presenting the result, we first present the definition of Harary network in the following. Recall that for a network containing $n$ nodes being resistant to $k$ link attacks, one necessary condition is that each node should have a degree of at least $k+1$, yielding the total number of links more than $\ceil*{\frac{(k+1)n}{2}}$. Here, $\ceil*{\cdot}$ denotes the ceiling operator. Harary network below can achieve this bound.
\begin{definition}[Harary Network \cite{harary}]
In a network containing $n$ nodes, Harary network is the optimal design that uses the minimum number of links equaling $\ceil*{\frac{(k+1)n}{2}}$ for the network still being connected after removing any $k$ links. 
\end{definition}
The constructive method of general Harary network can be described with cycles as follows. It first creates the links between node $i$ and node $j$ such that $(|i-j|\mod n)=1$, and then $(|i-j|\mod n)=2$, etc. When the number of nodes is odd, then the last cycle of link creation is slightly different since $\frac{(k+1)n}{2}$ is not an integer. However, the bound $\ceil*{\frac{(k+1)n}{2}}$ can be still be achieved.  For clarity, we illustrate three cases in Fig. \ref{Harary_figs} with $n=5,7$ under different security levels $k=2,3$. Since Harary network achieves the bound $\ceil*{\frac{(k+1)n}{2}}$, its computational cost of the construction is linear in both the number of nodes $n$ and the security level $k$.

\begin{figure}[!t]
\begin{centering}
\includegraphics[width=1\columnwidth]{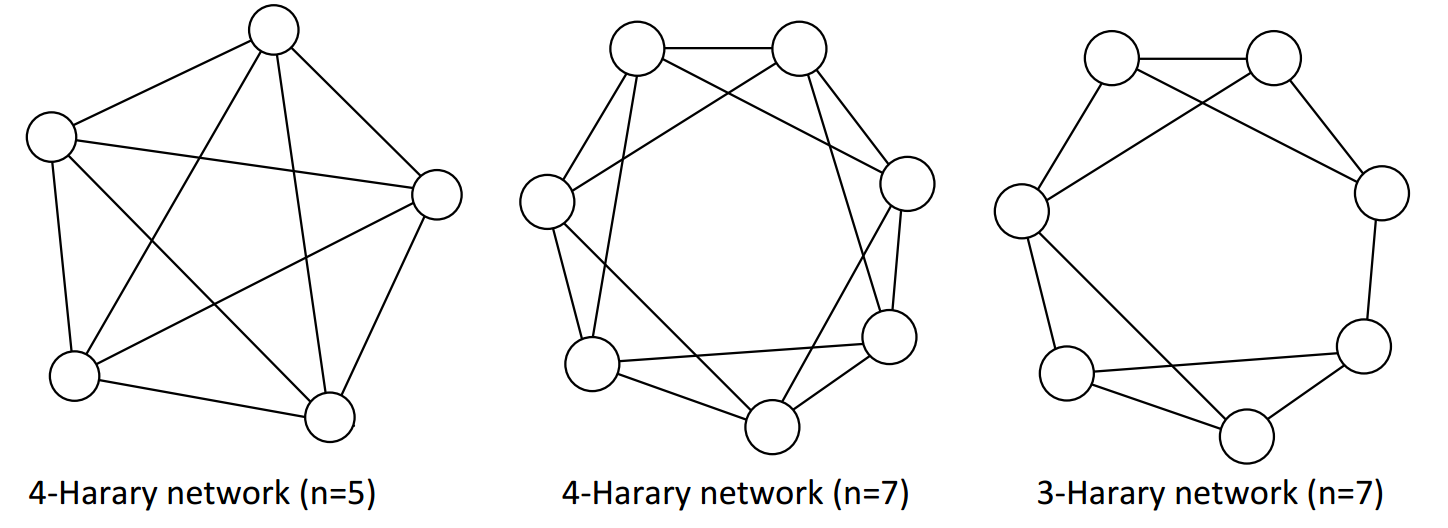} 
\par\end{centering}
\caption{\label{Harary_figs}
Illustration of Harary networks with different number of nodes and security levels.}
\end{figure}
 
Then, we obtain the following result.

\begin{proposition}\label{prop:1}
For the number of protected links $p$ taking values of $n-1,\ n_2-1$, and $0$, we successively have:
\begin{itemize}
\item[(i)] each $s^D_{n-1}$ contains exactly $0$ non-protected link.
\item[(ii)] each $s^D_{n_2-1}$ contains exactly $\ceil*{\frac{(n_1+1)(k_1+1)}{2}}$ non-protected links if and only if $k_1+1 \leq n_1$. 
\item[(iii)] if we have the following asumptions: (i) $k_1 \mod 2 = 1$, where $\mod$ denotes the modulus operator, (ii) $n_2 > k_2-k_1$ and (iii) $n_2\frac{k_1+1}{2} \leq n_1$, then each $s^D_{0}$ contains exactly $\displaystyle \ceil*{\frac{n_1 (k_1+1) + n_2(k_2+1)}{2}}$ non-protected links.
\end{itemize}
\end{proposition}

\begin{proof}
We successively prove the three items in the proposition in the following.

(i) Note that $s^D_{n-1}$ contains exactly $p=n-1$ protected links. It is thus possible to construct a tree network among the set $\set{1} \cup \set{2}$ of nodes  that consists of only  protected links. Thus, no non-protected link is required, and the lower bound (point E in Fig. \ref{fig:illus}) can be reached.

(ii) Suppose that $p=n_2-1$. If $k_1+1 \leq n_1$, we can construct any tree protected network on the nodes of \set{2}. Further, construct a $(k_1+1)$-Harary network on the nodes of $\set{1}\cup \{n_1+1\}$, that is the nodes of type 1 and one node of type 2. Such construction is possible since $k_1+2 \leq n_1+1$. The total number of non-protected links is then exactly $\ceil*{\frac{(n_1+1)(k_1+1)}{2}}$ (point C in Fig. \ref{fig:illus}). Therefore, each node in $\set{1}\cup \{n_1+1\}$ is connected to $k_1+1$ other nodes, and the IoT network cannot be disconnected after removing $k_1$ non-protected links. In addition, the subnetwork 2 is resistant to any number of attack since it is constructed using all protected links. Note that the constructed Harary network here is optimal, in the sense that its configuration uses the least number of links for the IoT network being (\crit{1}, \crit{2})-resistant.

Next, if $k_1+1> n_1$, then suppose that a network $g$ achieves the lower bound $\ceil*{\frac{(n_1+1)(k_1+1)}{2}}$. Consider its associated contracted network $\hat{g}$. Since $g$ contains $n_2-1$ protected links, then $\hat{g}$ is such that $\nu_0+\nu_1+\nu_2 \geq n_1+1$. From the shape of the lower bound $\Phi$ in the proof of Proposition  \ref{prop:lower}, then necessarily $\nu_0+\nu_2 = 1$ and $\nu_1=n_1$. Thus, all nodes in $\set{2}$ need to be connected together by protected links. Since $|\set{2}| = n_2$, then it requires at least $n_2-1$ protected links, which equals $p$. Thus, there cannot be any protected link involving nodes in set \set{1}. In addition, each node in \set{1} needs to be connected to at least $k_1+1$ other nodes in the IoT network. Since $k_1+1> n_1$, then every node in $\set{1}$ should connect to at least $(k_1+1)-(n_1-1) \geq 2$ number of nodes in $\set{2}$.  Recall that in a complete network of $m$ nodes, each node has a degree of $m-1$, and the total number of links is $\frac{m(m-1)}{2}$. Hence, our IoT network admits a completed graph in \set{1} with some extra $n_1((k_1+1)-(n_1-1))$ non-protected links between two subnetworks, and in total at least $\frac{n_1(n_1-1)}{2} + n_1((k_1+1)-(n_1-1)) = n_1(k_1+1)- \frac{n_1(n_1-1)}{2}$ non-protected links. Then, comparing with the lower bound, the extra number of links required is
$
n_1(k_1+1)- \frac{n_1(n_1-1)}{2}-\frac{(n_1+1)(k_1+1)}{2}
= \frac{n_1-1}{2} (k_1+1-n_1) >0.
$
Thus, $s^D_{n_2-1}$ does not achieve the lower bound (point C in Fig. \ref{fig:illus}) when $k_1+1>n_1$.

(iii) Finally, suppose that $p=0$. We renumber the nodes in the network according to the following sequence:
$
1, 2, \cdots, \frac{k_1+1}{2}, {{n_2}}, 
\frac{k_1+1}{2}+1, \cdots, k_1+1, {{n_2+1}}, 
k_1+2, \cdots, 3\frac{k_1+1}{2}, {{n_2+2}}, \cdots.$ Intuitively, we interpose one node in \set{2} after every $\frac{k_1+1}{2}$ nodes in \set{1}.
Then, we first build a $(k_1+1)$-Harary network among all the nodes in \set{1} and \set{2}. Note that since $n_2\frac{k_2+1}{2} \leq n_1$, then the last $\frac{k_1+1}{2}$ indices of the sequence only contain nodes of type $1$. Thus, by construction, there are no links between any two nodes in \set{2}. Then, we can further  construct a $(k_2-k_1)$-Harary network on the nodes in \set{2}, which is possible since $n_2 > k_2-k_1$. Thus, the constructed IoT network is (\crit{1}, \crit{2})-resistant, and it is also optimal since it uses the minimum number of non-protected links.
\end{proof}

Proposition \ref{prop:1} and Fig.~\ref{fig:illus} indicate  that depending on the system parameters ($k_1,k_2,n_1,n_2$) and for a given budget, the optimal IoT network can achieve at either point A, C or E with $p=0,n_2-1,n-1$ protected links, respectively. Notice that when $k_1+1>n_1$, $s^D_{n_2-1}$ is not optimal at point C and the lower bound on the number of non-protected links is not attained. Instead, in this case, $s^D_{n_2-1}$ requires $\frac{n_1(2k_1-n_1+3)}{2}$ non-protected links in which $n_1(k_1-n_1+2)$ are allocated between two subnetworks, introducing protection redundancy for nodes in \set{2}. For the IoT network containing 0 protected link, it reaches the lower bound (point A) if we can construct a $(k_1+1)$-Harary network for all nodes and an additional $(k_2-k_1)$-Harary network for nodes only in \set{2}. As mentioned before, the Harary network admits an optimal configuration with the maximum connectivity given a number of links \cite{harary}.

\subsection{Optimal Strategy and Construction of IoT Networks}
\label{sec:theo}
We investigate the optimal strategy and the corresponding construction for the IoT network designer in this section. 
\subsubsection{Optimal Strategy}
Before presenting the main result, we comment on the scenarios that we aim to study regarding the IoT networks.
\begin{itemize}
\item[(1)] First, the number of nodes is relatively large comparing with the link failure risks, i.e., $n_1 \geq k_1+1$ and $n_2 \geq k_2-k_1+1$. Indeed, these two conditions indicate that the designer can create a secure two-layer IoT network solely using non-protected links.
\item[(2)] We further have the condition $n_2\frac{k_1+1}{2} \leq n_1$, indicating that the type $2$ nodes with higher criticality levels in \set{2} constitute a relatively small portion in the IoT network comparing with these in \set{1}. This condition also aligns with the practice that the attacker has preferences on the nodes to compromise in the IoT which generally only contain a small subset of the entire network.
\item[(3)] Finally, we have constraints $k_1 \mod 2 = 1$ and $n_2 (k_2+1) \mod 2 = 0$ which are only used to simplify the presentation of the paper (whether the number of nodes and attacks is odd or even). However, they do not affect the results significantly. Note that different cases corresponding to $k_1 \mod 2 = 0$ or $n_2 (k_2+1) \mod 2 = 1$ can be studied in a similar fashion as in our current context. The only difference is that for certain system parameters, $s^D_{0}$ is not an optimal strategy comparing with $s^D_{n_2-1}$  by following a similar analysis in~\cite{Christophe}.
\end{itemize} 

Therefore, based on the above conditions, the scenarios that we analyze are quite general and conform with the situations in the adversarial IoT networks. Based on Proposition~\ref{prop:1}, we then obtain the following result on the optimal design of secure two-layer IoT networks. Note that the solution in Proposition 3 is optimal to the original optimization problem presented in Section II under the considered scenarios.

\begin{proposition}\label{prop:theo}
Under the conditions that $n_1 \geq k_1+1$, $n_2 \geq k_2-k_1+1$, $n_2\frac{k_1+1}{2} \leq n_1$, $k_1 \mod 2 = 1$ and $n_2 (k_2+1) \mod 2 = 0$, we have the following results:
\begin{itemize}
\item[I)] \textit{Regime I:} if $1+k_1-n(k_2-k_1)\leq 0$, then:
\begin{itemize}
\item[(1)] if $2 \frac{c_P}{c_{NP}} \geq k_2+1+\frac{k_2-k_1}{n_2-1}$, then $s^D_0$ are optimal strategies.
\item[(2)] if $k_1+1+\frac{k_1+1}{n_1} \leq   2 \frac{c_P}{c_{NP}} < k_2+1+\frac{k_2-k_1}{n_2-1}$, then $s^D_{n_2-1}$ are optimal strategies.
\item[(3)] if $2 \frac{c_P}{c_{NP}} < k_1+1+\frac{k_1+1}{n_1}$, then $s^D_{n-1}$ are optimal strategies.
\end{itemize}
\item[II)] \textit{Regime II:} if $1+k_1-n(k_2-k_1)> 0$, then:
\begin{itemize}
\item[(1)]  when $k_2-k_1+1\leq n_2<\frac{1+k_1}{1+k_1-n_1(k_2-k_1)}$, the optimal IoT network design strategies are the same as those in regime I.
\item[(2)] otherwise, i.e., $n_2\geq\frac{1+k_1}{1+k_1-n_1(k_2-k_1)}$, we obtain
\begin{itemize}
\item[(i)] if $2\frac{c_P}{c_{NP}}\geq \frac{n_1(k_1+1)+n_2(k_2+1)}{n_1+n_2-1}$, then $s^D_0$ are optimal strategies.
\item[(ii)] if $2\frac{c_P}{c_{NP}}< \frac{n_1(k_1+1)+n_2(k_2+1)}{n_1+n_2-1}$, then $s^D_{n-1}$ are optimal strategies.
\end{itemize}
Thus, $s^D_{n_2-1}$ cannot be optimal in this scenario.
\end{itemize}
\end{itemize}
\end{proposition}

\begin{proof}
From Proposition~\ref{prop:1} and under the assumptions in the current proposition, $s^D_0$, $s^D_{n_2-1}$ and $s^D_{n-1}$ achieve the lower bounds of the number of links for the network being $(\crit{1}, \crit{2})$-resistant. 
In Fig.~\ref{fig:illus}, note that the slope of the line between points A and C is $\frac{1}{2} (k_2+1+\frac{k_2-k_1}{n_2-1} )$, and between points C and E is $\frac{1}{2} ( k_1+1+\frac{k_1+1}{n_1} )$, where we quantify the slopes in their absolute value sense. 

In regime I, i.e., $1+k_1-n(k_2-k_1)\leq 0$, we obtain $(k_2+\frac{k_2-k_1}{n_2-1} ) - ( k_1+\frac{k_1+1}{n_1} )\leq 0$, yielding that  
 the line connecting points A and C has a higher slope than the one joining points C and E. Thus, if the lines of iso-costs have a slope higher than the slope of the line A-C, then the minimum cost is obtained at point A. Similarly, if the slope is less than that of line C-E, then the minimum cost is obtained at point E. Otherwise, the minimum is obtained at point C. Recall that the slope of the lines of iso-costs is equal to $c_P/c_{NP}$ which leading to the result.

In the other regime II, i.e., $1+k_1-n(k_2-k_1)>0$, the slope of line A-C is not always greater than that of line C-E. Specifically, we obtain a threshold $n_2=\frac{1+k_1}{1+k_1-n_1(k_2-k_1)}$ over which the slop of line C-E is greater than line A-C. Therefore, if $n_2<\frac{1+k_1}{1+k_1-n_1(k_2-k_1)}$, the optimal network design is the same as those in regime I. In addition, when $n_2\geq\frac{1+k_1}{1+k_1-n_1(k_2-k_1)}$, and if the slop of iso-costs lines, i.e., $c_P/c_{NP}$, is larger than the slope of the line connecting points A and E, the minimum cost is achieved at point A. Otherwise, if  $c_P/c_{NP}$ is smaller than the slop of line A-E, the optimal network configuration is obtained at point E. 
\end{proof}

From Proposition \ref{prop:theo}, we can conclude that in regime I, i.e., $1+k_1-n(k_2-k_1)\leq 0$, when the unit cost of protected links is relatively larger than the non-protected ones, then the secure IoT networks admit an $s^D_0$ strategy using all non-protected links. In comparison, the secure IoT networks are constructed with solely protected links when the cost per protected link is relatively small satisfying $ {c_P} < (k_1+1+\frac{k_1+1}{n_1})c_{NP}/2$. Note that the optimal network design strategy in this regime can be achieved by protecting the minimum spanning tree for a connected network. Equivalently speaking, finding a spanning tree method provides an algorithmic approach to construct the optimal network in this regime. Finally, when the cost per protected link is intermediate, the network designer allocates $n_2-1$ protected links connecting those critical nodes in set \set{2} while uses non-protected links to connect the nodes in \set{1}. In addition, the intralinks between two subnetworks are non-protected ones.

Note that the specific configuration of the optimal IoT network is not unique according to Proposition \ref{prop:theo}. To enhance the system reliability and efficiency, the network designer can choose the one among all the optimal topology that minimize the communication distance between devices.

Since the cyber threat in subnetwork 2 is more severe than that in subnetwork 1, i.e., $k_2\geq k_1$, thus the condition of regime II in Proposition \ref{prop:theo} ($1+k_1-n(k_2-k_1)> 0$) is not generally satisfied. We further have the following Corollary refining the result of optimal IoT network design in regime II.

\begin{corollary}\label{coro2}
Only when two subnetworks facing the same level of cyber threats, i.e., $k_1=k_2$, the optimal IoT network design follows the strategies in regime II. Moreover, $s^D_{n_2-1}$ cannot be an optimal network design in regime II.
\end{corollary}
\begin{proof}
Based on the condition $n_1\geq k_1+1$, we obtain $1+k_1-(n_1+n_2)(k_2-k_1)\leq n_1-(n_1+n_2)(k_2-k_1)$. Thus, when $k_2>k_1$, the condition of regime II ($1+k_1-n(k_2-k_1)> 0$) cannot be satisfied. Since $k_2\geq k_1$, then only $k_1=k_2$ yields $1+k_1>0$. Therefore, $n_2\geq\frac{1+k_1}{1+k_1-n_1(k_2-k_1)}=1$ always holds which leads to the result.
\end{proof}

 We then simplify the conditions leading to regime I and II as follows.
 
\begin{corollary}\label{coro3}
The IoT network design can be divided into two regimes according to the cyber threat levels. Specifically, when $k_2>k_1$, the optimal design strategy follows the one in regime I in Proposition \ref{prop:theo}, and otherwise ($k_1=k_2$) follows the one in regime II.
\end{corollary}

We illustrate the optimal design strategies in Fig. \ref{regime_optimal} according to the heterogeneous security requirements and link creation costs ratio.

\begin{figure}[!t]
\begin{centering}
\includegraphics[width=1\columnwidth]{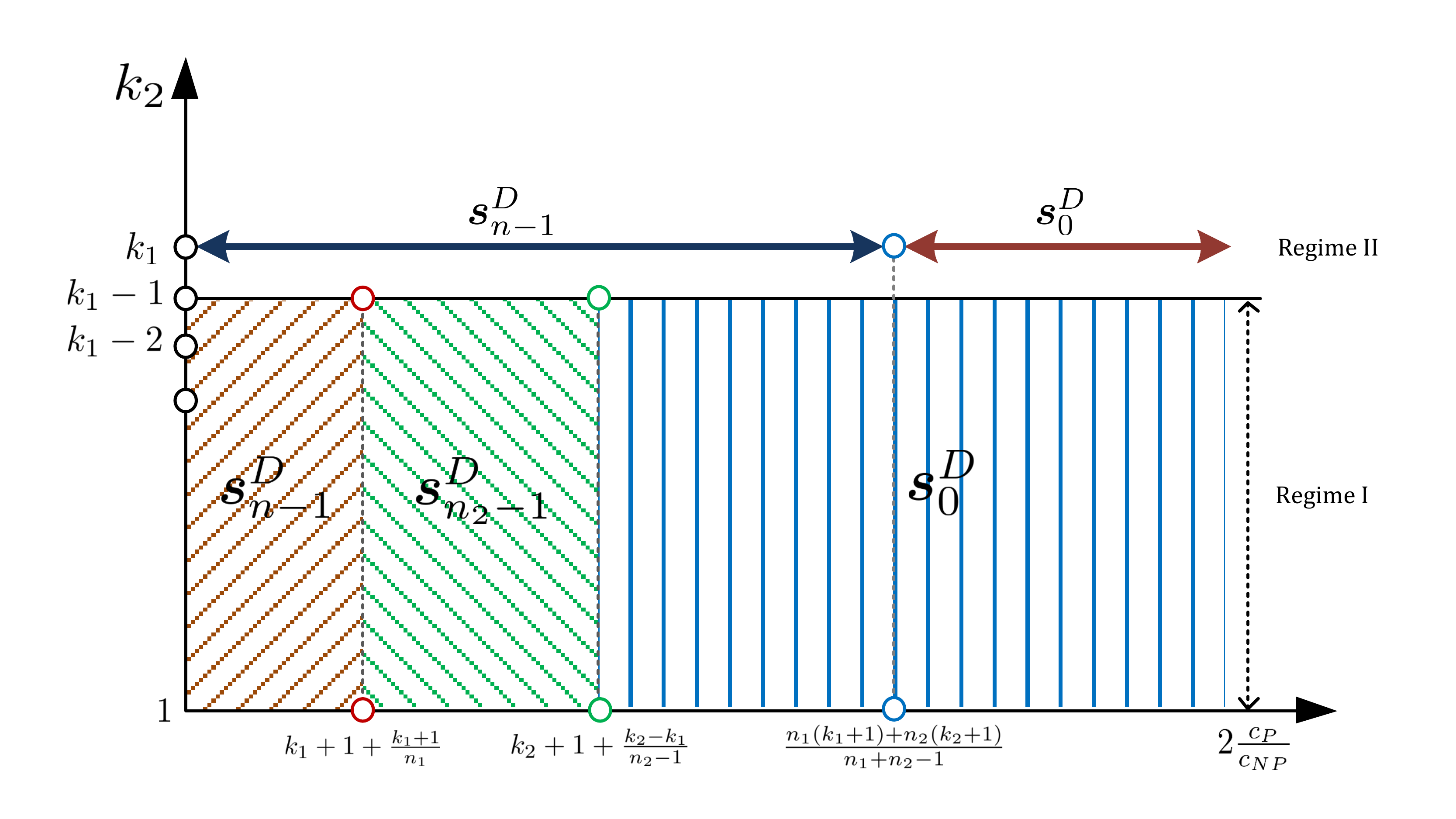} 
\par\end{centering}
\caption{\label{regime_optimal}
Optimal design of two-layer IoT networks in two regimes in terms of system parameters. When $k_2>k_1$, the optimal network design follows from the strategies in regime I which can be in any $s_{n-1}^D$, $s_{n_2-1}^D$ or $s^D_{0}$ depending on the value of $\frac{c_P}{c_{NP}}$. When $k_2=k_1$, the IoT network designer chooses strategies from regime II, either of $s_{n-1}^D$ or $s^D_{0}$ in term of the  link cost ratio $\frac{c_P}{c_{NP}}$.}
\end{figure}

\subsubsection{Robust Optimal Strategy}
One interesting phenomenon is that some strategies are optimal for a class of security requirements. Thus, these strategies are robust in spite of the dynamics of cyber threat levels. We summarize the results in the following Corollary.

\begin{corollary}\label{coro4}
Consider to design a $(k_1,k_2)$-resistant IoT network. If $s^D_{n-1}$ is the optimal strategy, then it is robust and optimal to security requirement for the network being $(k_1',k_2')$-resistant, for all $k_1'>k_1$ and all $k_2'>k_2$. If $s^D_{n_2-1}$ is the optimal strategy, then it is robust and optimal to cyber threat levels $(k_1,k_2')$, for all $k_2'>k_2$. Furthermore, the optimal strategy $s^D_0$ is not robust to any other security standards $(k_1',k_2')$, for $k_1'\neq k_1$ and $k_2'\neq k_2$.
\end{corollary}

Corollary \ref{coro4} has a natural understanding on the selection of robust strategies. When the cyber threat level increases, then the optimal network $s^D_{n-1}$ remains to be optimal since the network construction cost does not increase under $s^D_{n-1}$. Under the optimal $s^D_{n_2-1}$, subnetwork 2 is connected with all protected links and the rest is connected by a Harary network with the minimum cost. If subnetwork 2 faces more attacks, ($k_2$ becomes larger), then $s^D_{n_2-1}$ is robust and optimal in the sense that subnetwork 2 remains secure and no other non-protected link is required.

\subsubsection{Construction of the Optimal Secure IoT Networks} 
We present the constructive methods of optimal IoT networks with parameters in different regimes based on Proposition \ref{prop:theo}. 

Specifically, the optimal $s^D_{n-1}$ can be constructed by any tree network using protected links. In addition, the optimal networks $s^D_{n_2-1}$ can be constructed in two steps as follows. First, we create a tree protected network on the nodes of \set{2}. Then, we construct a $(k_1+1)$-Harary network on the nodes of $\set{1}\cup \{n_1+1\}$, i.e., all nodes of type 1 and one node of type 2, where a constructive method of Harary network can be found in~\cite{harary}. 

Finally, regarding the optimal network $s^D_{0}$, we build it with the following procedure. First, we renumber the nodes  according to the sequence:
$
1, 2, ..., \frac{k_1+1}{2}, { {n_2}}, 
\frac{k_1+1}{2}+1, ..., k_1+1, { {n_2+1}}, 
k_1+2, ..., 3\frac{k_1+1}{2}, { {n_2+2}}, ...
$ Recall that this renumbering sequence can be achieved by interpolating one node in \set{2} after every $\frac{k_1+1}{2}$ nodes in \set{1}.
Then, we build a $(k_1+1)$-Harary network among all the nodes in \set{1} and \set{2}. Finally, we construct a $(k_2-k_1)$-Harary network on the nodes in \set{2}.

\subsubsection{Consideration of Random Link Failures}
In the considered model so far, the non-protected communication link between nodes is removed with probability 1 by the attack and remains connected without attack. In general, the non-protected links face random natural failures.
If we consider this random failure factor, then there is a probability that the designed optimal network will be disconnected under the joint cyber attacks and failures. We assume perfect connection of protected links and denote the random failure probability of a non-protected link by $\kappa\in[0,1)$.  Therefore, in the regime that the optimal network design is of Harary network where all links are non-protected, then under the anticipated level of cyber attacks, a single link failure of non-protected link will result in the network disconnection. Thus, the probability of network connection, i.e., mean connectivity, is equal to $(1-\kappa)^{\ceil*{\frac{n_1(k_1+1)+n_2(k_2+1)}{2}}-k_2}\approx (1-\kappa)^{\frac{n_1(k_1+1)+n_2(k_2+1)-2k_2}{2}}$ which is of order $(1-\kappa)^{\frac{n_1k_1+n_2k_2}{2}}$. Similarly, under the regime that the optimal network admits $n_2-1$ protected links and $\ceil*{\frac{(k_1+1)(n_1+1)}{2}}$ non-protected links, the probability of network connection under link failure is $(1-\kappa)^{\ceil*{\frac{(k_1+1)(n_1+1)}{2}}}\approx (1-\kappa)^{\frac{(k_1+1)(n_1+1)}{2}}$ which is of order $(1-\kappa)^{\frac{k_1n_1}{2}}$. We can see that in the above two regimes, when the security requirement is not relatively high and the size of the network is not large, the current designed optimal strategy gives a relatively high mean network connectivity.  In the regime that the optimal network is constructed with all protected links, then the mean network connectivity is 1 where the random failure effect is removed.

\section{Case Studies} \label{case_study}
In this section, we use case studies of IoBT to illustrate the optimal design principals of secure networks with heterogeneous components. The results in this section are also applicable to other mission-critical IoT network applications.

The IoBT network designer determines the optimal strategy on creating links with/without protection between agents in the battlefield. The ground layer and aerial layer in IoBT generally face different levels of cyber threats which aim to disrupt the network communications. Since UAVs become more powerful in the military tasks, they are the primal targets of the attackers, and hence the UAV network faces an increasing number of cyber threats.
 In the following case studies, we investigate the scenario that the IoBT network designer anticipates more cyber attacks on the UAV network than the soldier and UGV networks. 
  The cost ratio between forming a protected link and a unprotected link $\frac{c_p}{c_{NP}}$ is critical in designing the optimal IoBT network. This ratio depends on the number of channels used in creating a safe link though MTD. We will analyze various cases in the following studies.

\subsection{Optimal IoBT Network Design} \label{case1}

Consider an IoBT network consisting of $n_1=20$ soldiers and $n_2=5$ UAVs ($n=25$). The designer aims to design the ground network and the UAV network resistant to $k_1=5$ and $k_2=9$ attacks, respectively. Hence the global IoBT network is $(5,9)$-resistant. Based on Proposition \ref{prop:theo}, the system parameters satisfy the condition of regime I. Further, we have two critical points $T_1 := (k_1+1+\frac{k_1+1}{n_1})/2=3.15$ and $T_2:=(k_2+1+\frac{k_2-k_1}{n_2-1} )/2= 5.5$, at which the topology of optimal IoBT network encounters a switching. For example, when a protected link adopts 3 channels to prevent from attacks, i.e., $\frac{c_p}{c_{NP}}=3$, the optimal IoBT network is an $s^D_{24}$ graph as shown in Fig. \ref{case_A}\subref{case_A_1}. When a protected link requires 5 channels to be perfectly secure, i.e., $\frac{c_p}{c_{NP}}=5$, then the optimal IoBT network is of $s^D_{4}$ configuration which is depicted in Fig. \ref{case_A}\subref{case_A_2}. In addition, if the cyber attacks are difficult to defend against (e.g., require 7 channels to keep a link safe, i.e., $\frac{c_p}{c_{NP}}=7$), the optimal IoBT network becomes an $s^D_{0}$ graph as shown in Fig. \ref{case_A}\subref{case_A_3}. The above three types of optimal networks indicate that the smaller the cost of a protected link is, the more secure connections are formed starting from the UAV network to the ground network. 

\begin{figure}[!t]
  \centering
  \subfloat[$s^D_{24}$ IoBT network]{
    \includegraphics[width=1\columnwidth]{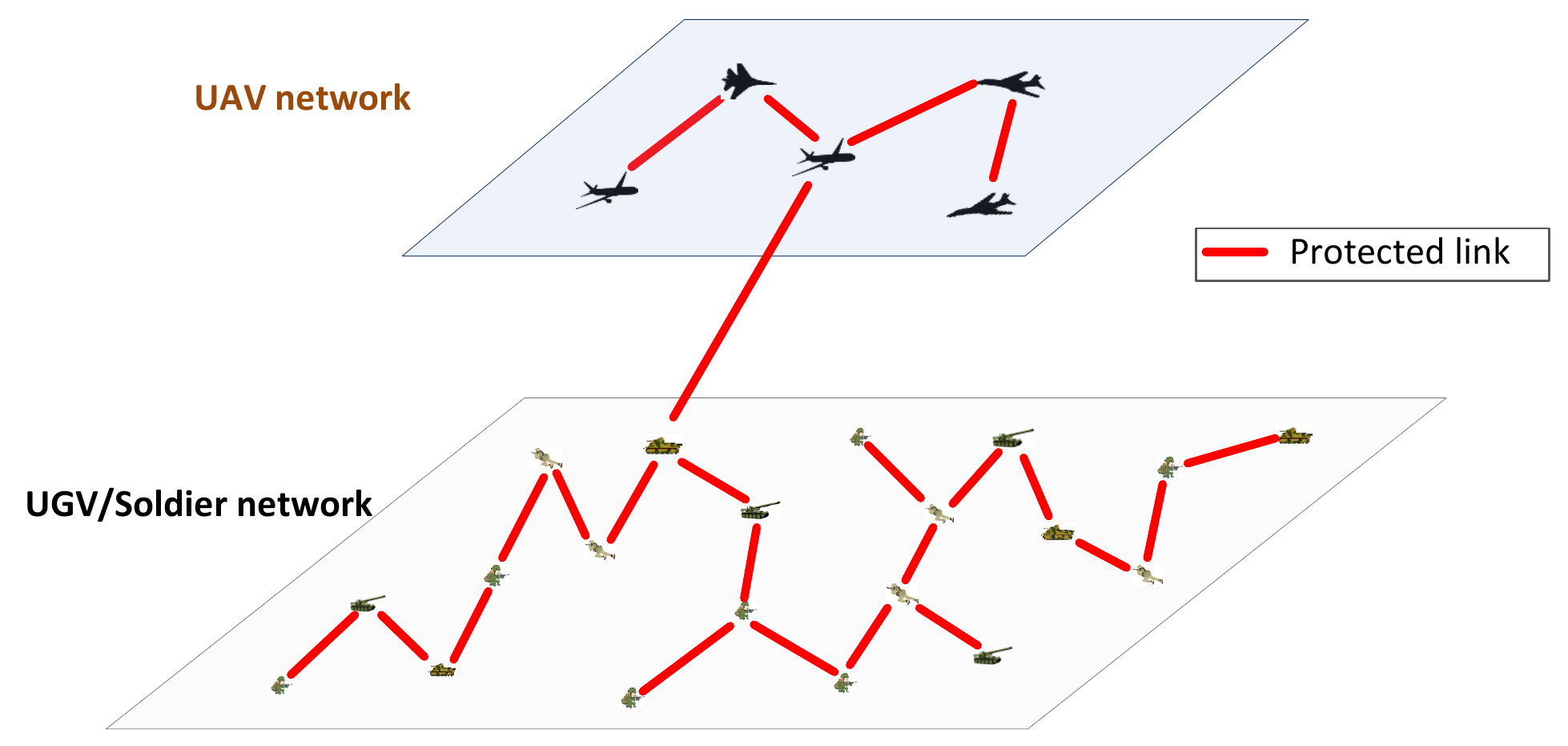}\label{case_A_1}}\\ 
	 \subfloat[$s^D_{4}$ IoBT network]{
    \includegraphics[width=1\columnwidth]{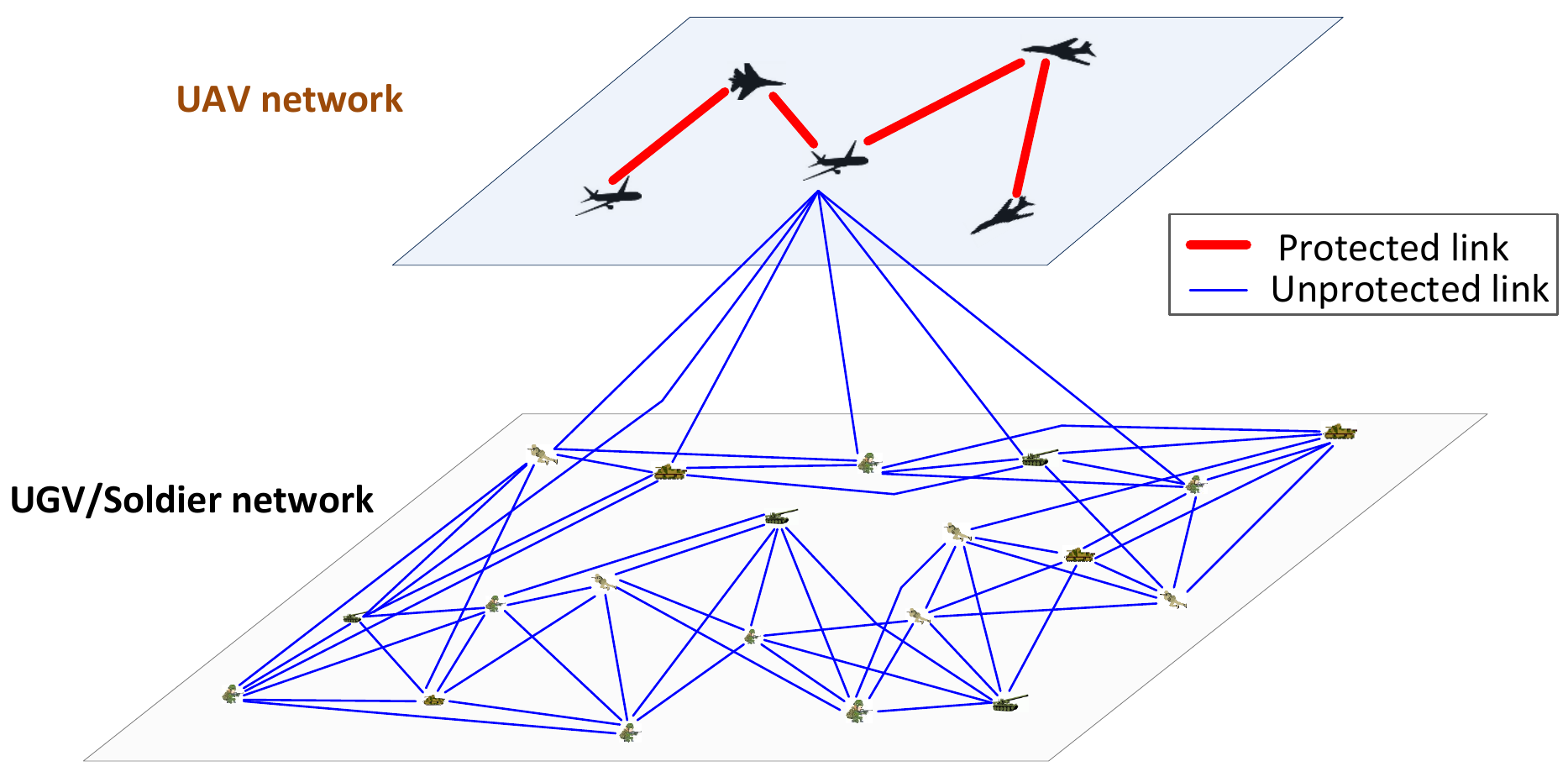}\label{case_A_2}} \\ 
      \subfloat[$s^D_{0}$ IoBT network]{
    \includegraphics[width=1\columnwidth]{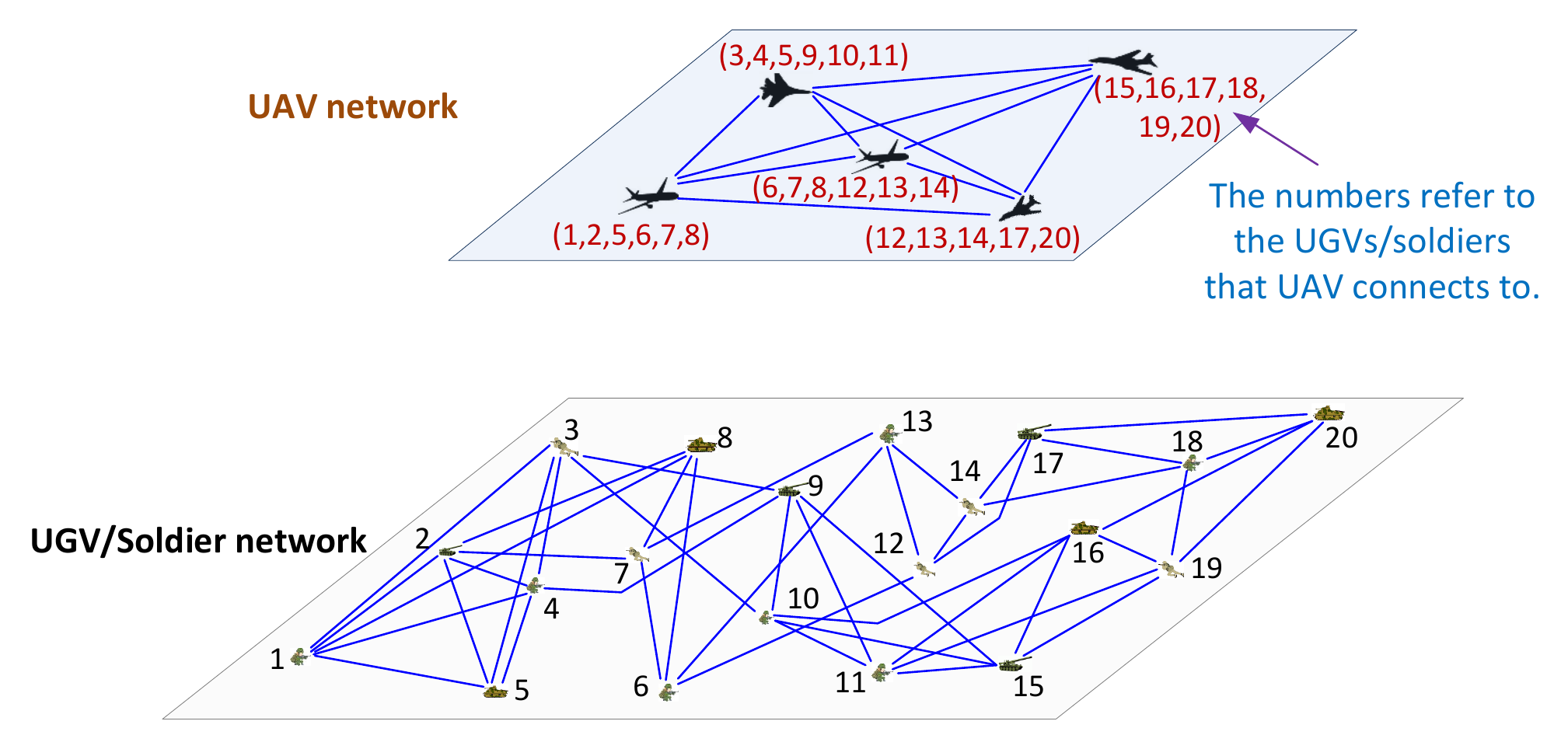}\label{case_A_3}} 
  \caption[]{ (a) When $\frac{c_p}{c_{NP}}=3<T_1$, the optimal IoBT network is an $s^D_{24}$ graph with all protected links. (b) When $T_1<\frac{c_p}{c_{NP}}=5<T_2$, the optimal network is an $s^D_{4}$ graph, where the UAV network is connected with protected links and the ground network with all unprotected links. (c) When $\frac{c_p}{c_{NP}}=7>T_3$, the optimal IoBT network adopts an $s^D_{0}$ configuration with all unprotected links.} 
  \label{case_A}
\end{figure}

\subsection{Resilience of the IoBT Network}

The numbers of UAVs, UGVs and soldiers can be dynamically changing. To study the resilience of the designed network, we first investigate the scenario that a number of UGVs/soldiers join the battlefield which can be seen as army backups. As $n_1$ increases, the threshold $T_1$ decreases slightly while $T_2$ remains unchanged. Therefore, the optimal IoBT network keeps with a similar topology except that the newly joined UGVs/soldiers connect to a set of their neighbors. To illustrate this scenario, we present the optimal network with $n_1=22$ and $\frac{c_p}{c_{NP}}=5$ in Fig. \ref{case_B}\subref{case_B_1}, and all the other parameters stay the same as those in Section \ref{case1}. When $n_1$ decreases, the network remains almost unchanged except those UGVs/soldiers losing communication links build up new connections with neighbors. An illustrative example with $n_1=17$ is depicted in Fig. \ref{case_B}\subref{case_B_2}.

\begin{figure}[!t]
  \centering
  \subfloat[]{
    \includegraphics[width=1\columnwidth]{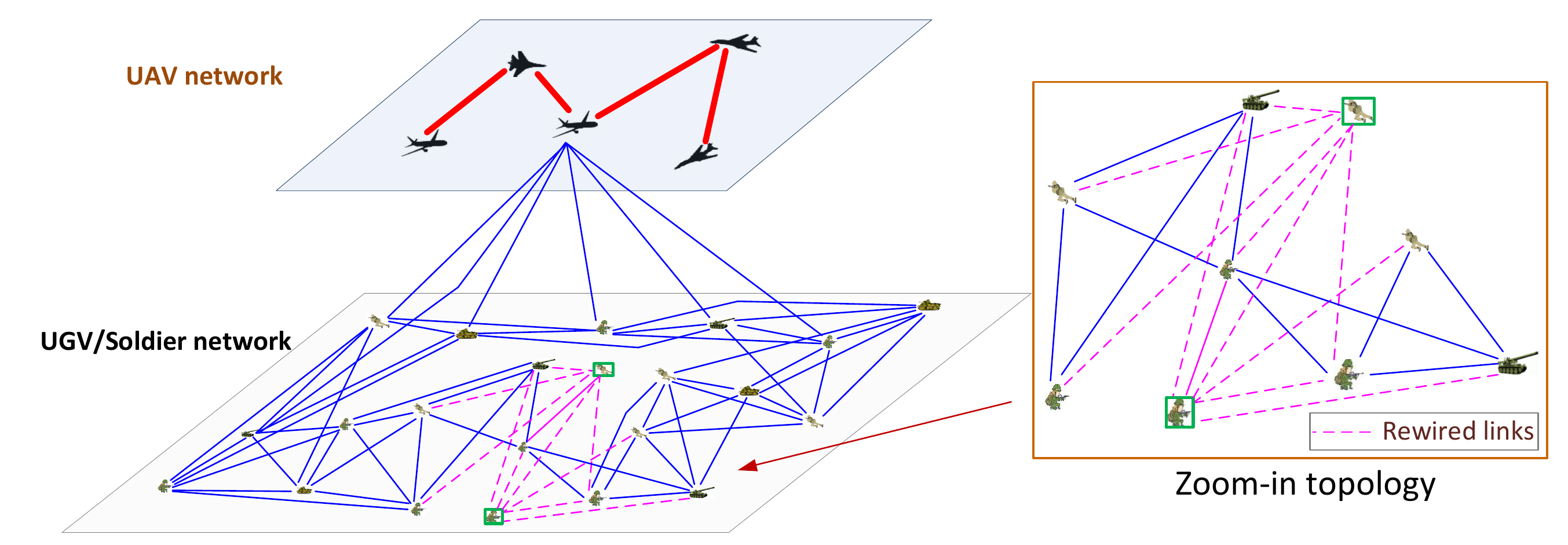}\label{case_B_1}}\\ 
	 \subfloat[]{
    \includegraphics[width=1\columnwidth]{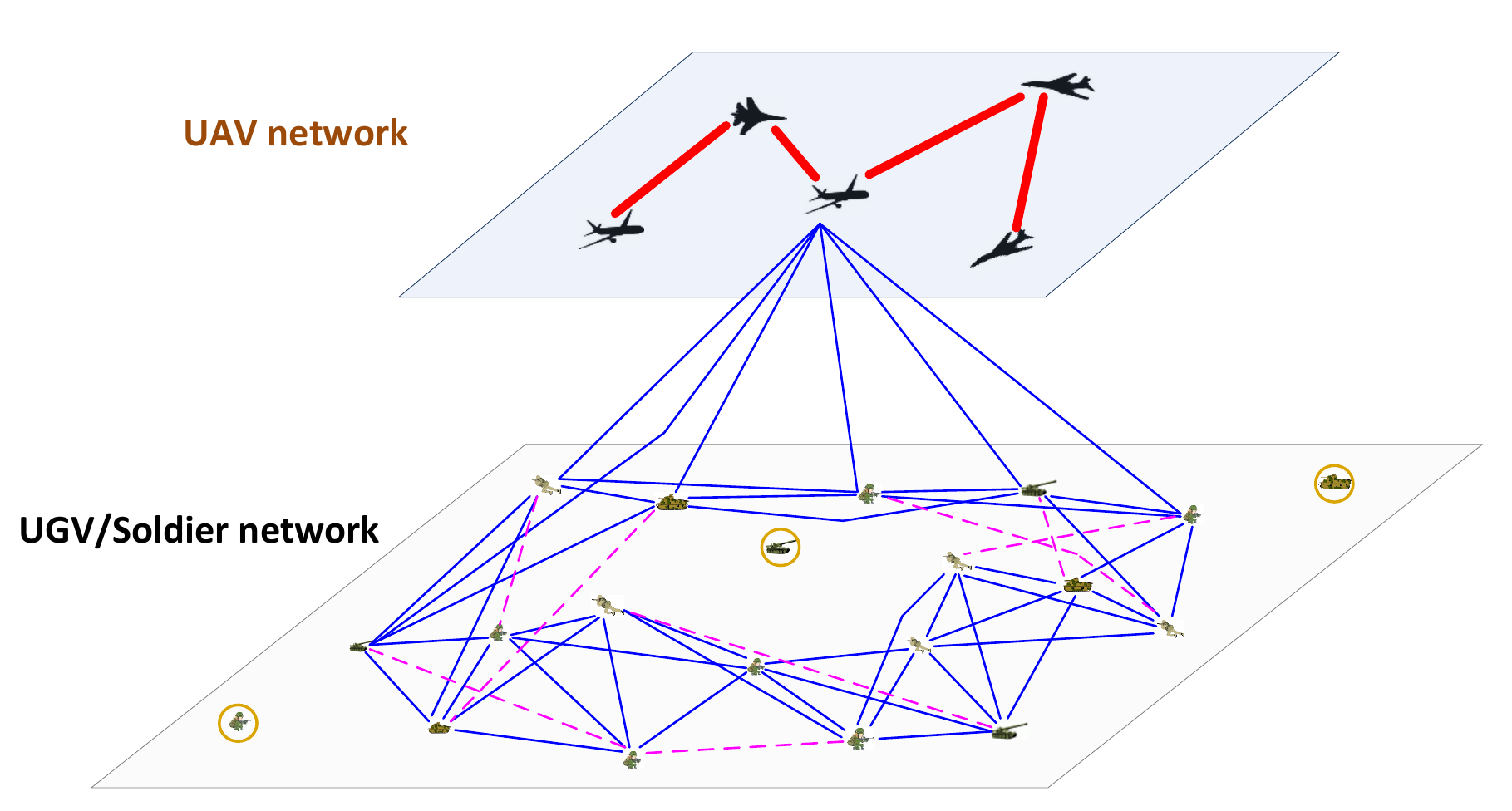}\label{case_B_2}} 
  \caption[]{(a) and (b) show the optimal IoBT network reconfiguration when two UGVs/soldiers join in and leave the battlefield, respectively.}
  \label{case_B}
\end{figure}

Another interesting scenario is that when the number of UAVs $n_2$ changes due to backup aerial vehicles joining in and current vehicles leaving the battlefield for maintenance. When $n_2$ increases, then the threshold $T_1$ remains the same while $T_2$ decreases. If the cost ratio $\frac{c_p}{c_{NP}}$ lies in the same regime with respect to $T_1$ and $T_2$ even though $T_2$ decreases, then under $\frac{c_p}{c_{NP}}\leq T_2$, the newly joined UAV will connect with another UAV with a protected link which either creates an $S^D_{n-1}$ or $s^D_{n_2-1}$ graph. Otherwise, if $\frac{c_p}{c_{NP}}>T_2$, the UAV first connects to other UAVs and then connects to a set of UGVs/soldiers both with unprotected links which yields an $s^D_0$ graph. When a number of UAVs leaving the battlefield, i.e., $n_2$, decreases, then $T_1$ stays the same and $T_2$ will increase under which the cost ratio $\frac{c_p}{c_{NP}}$ previous belonging to interval $\frac{c_p}{c_{NP}}\geq T_2$ may change to interval $T_1\leq \frac{c_p}{c_{NP}}\leq T_2$. Note that regime switching can also happen when $n_2$ increases. Therefore, the optimal IoBT network switches from $s^D_0$ to $s^D_{n_2-1}$ (for the increase of $n_2$ case, the switching is in a backward direction). For example, when the network contains $n_2=6$ UAVs and $\frac{c_p}{c_{NP}}=5.4$, and the other parameters are the same as those in Section \ref{case1}, from Proposition \ref{prop:theo}, the optimal IoBT network is an $s^D_{0}$ graph. However, Fig. \ref{case_A}\subref{case_A_2} shows that the optimal network adopts an $s^D_{4}$ topology when $n_2=5$. Therefore, by adding a UAV to the aerial layer, the optimal IoBT network switches from  $s^D_{4}$ to $s^D_0$ in this scenario. The interpretation is that a smaller number of UAVs is easier for the aerial network to defend against attacks, and hence protected links are used between UAVs instead of redundant unprotected links.

\subsection{Flexible Design and Robust Strategies}\label{robust_case}
In this section, we further investigate the secure IoBT network design in the presence of varying levels of cyber threats. Specifically, the parameters are selected as follows: $n_1=20$, $n_2 = 10,\ k_1=5$, and $\frac{c_P}{c_{NP}} = 5$. The security requirement $k_2$ takes a value varying from 5 to 14, modeling the dynamic or uncertain behaviors of the attacker targeting at the critical UAV network. The optimal IoBT network design is depicted in Fig. \ref{varying_k2}, and the corresponding cost is in shown Fig. \ref{total_cost}. When $k_2\in\llbracket 5, 8 \rrbracket$, the optimal IoBT network is constructed with all non-protected links. Since $k_2$ becomes larger, the number of non-protected links used is increasing, and thus the total cost increases. The optimal network topology switches from $s^D_0$ to $s_{9}^D$ when $k_2$ exceeds the threshold $8$. Then, when $k_2\in\llbracket 9, 14 \rrbracket$, the optimal IoBT network is unchanged as well as the associated construction cost. Despite the increases in $k_2$, no addition links are required since the UAV network (subnetwork 2) is connected with all protected links. Note that $s_{9}^D$ is a robust strategy in the sense that the IoBT network can be $(5,k_2)$-resistant, for all $k_2\in\llbracket 9, 14 \rrbracket$. This study can be generalized to the cases when the network designer has an uncertain belief on the attacker's strategy. Therefore, the IoBT designer can prepare for a number of attacking scenarios and choose from these designed strategies in the field with a timely and flexible manner.

\begin{figure}[!t]
\begin{centering}
\includegraphics[width=.95\columnwidth]{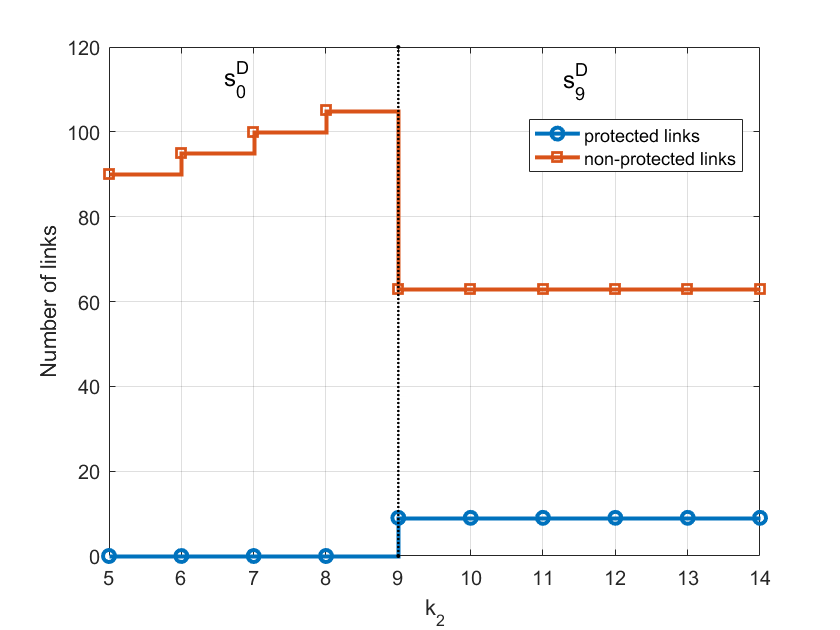}
\par\end{centering} 
\caption{\label{varying_k2}
Optimal IoBT network design with parameters $n_1=20$, $n_2 = 10,\ k_1=5$, $\frac{c_P}{c_{NP}} = 5$, and $k_2$ taking a value from 5 to 14. When $k_2\in\llbracket 5, 8 \rrbracket$, the optimal network design is in the form of $s_0^D$. When $k_2\in\llbracket 9, 14 \rrbracket$, the optimal network admits a strategy of $s_{9}^D$. Note that $s_{9}^D$ is robust to a dynamic or varying number of cyber attacks ranging from 9 to 14. }
\end{figure}

\begin{figure}[!t]
\begin{centering}
\includegraphics[width=.95\columnwidth]{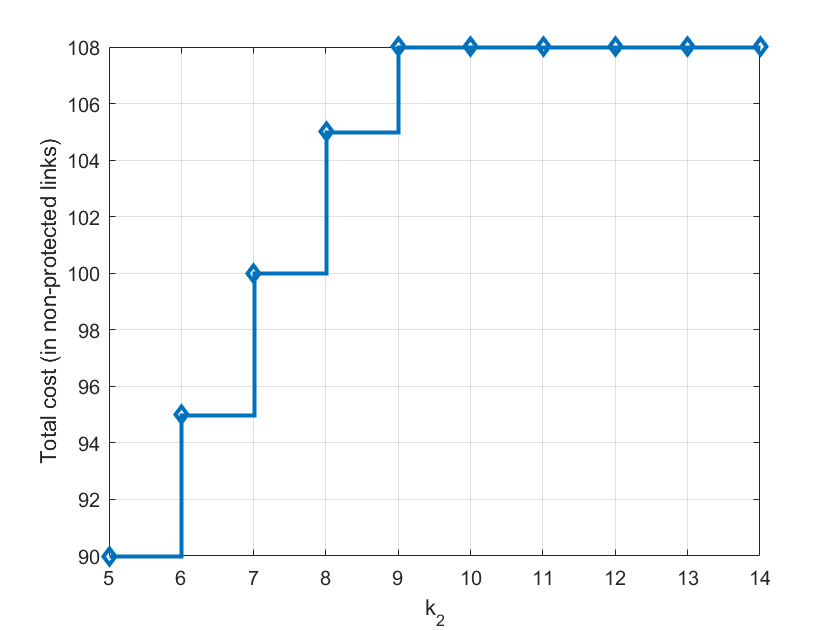}
\par\end{centering} 
\caption{\label{total_cost}
The total cost of optimal network design in terms of the number of non-protected links. In the regime of $k_2\in\llbracket 5, 8 \rrbracket$, with a larger $k_2$, the IoBT network requires more non-protected links to be resistant to attacks. In the regime of $k_2\in\llbracket 9, 14 \rrbracket$, the total cost remains the same, since  the UAV network is connected with all protected links and no additional non-protected link is required despite of the increasing cyber threats $k_2$.}
\end{figure}

\section{Conclusion}\label{conclusion}
In this work, we have studied a two-layer secure network formation problem for IoT networks in which the network designer aims to form a two-layer communication network with heterogeneous security requirements while minimizing the cost of using protected and unprotected links. We have shown a lower bound on the number of non-protected links of the optimal network and developed a method to construct networks that satisfy the heterogeneous network design specifications. We have demonstrated the design methodology in the IoBT networks. It has been shown that the optimal network can reconfigure itself adaptively as nodes enter or leave the system. In addition, the optimal IoBT network configuration may encounter a topological switching when the number of UAVs changes. We have further identified the optimal design strategies that can be robust to a set of security requirements. As part of the future work, we would extend the single network designer problem to a two-player one, where each player designs their own subnetwork in a decentralized way. Another direction will be generalizing the current bi-level network to more than two layers and designing the optimal strategies.

\bibliographystyle{IEEEtran}
\bibliography{IEEEabrv,references}

\end{document}